\documentclass[a4paper,10pt,reqno]{amsart}

\usepackage{hyperref}
\usepackage{amssymb}
\usepackage{eucal}
\usepackage{eufrak}
\usepackage{a4wide}
\usepackage{graphics}
\usepackage{graphicx}
\usepackage{epsfig}
\usepackage{multicol}
\usepackage{xypic}
\usepackage{amsmath}
\usepackage{wasysym}
\usepackage{dsfont}
\usepackage{float}
\usepackage{fancyhdr}
\usepackage{subfigure}

\usepackage{tikz}

\usepackage{algorithmicx}
 \usepackage{algorithm}
 \usepackage{algpseudocode}

\usepackage[T1]{fontenc}
\usepackage[utf8]{inputenc}

\usepackage{eufrak}
\usepackage{graphics}
\usepackage{graphicx}
\usepackage{epsfig}
\usepackage{multicol}
\usepackage{xypic}
\usepackage{amsmath}
\usepackage{color}

\newtheorem{proposition}{Proposition}[section]
\newtheorem{theorem}{Theorem}[section]
\newtheorem{lemma}{Lemma}[section]

\theoremstyle{definition}
\newtheorem{definition}{Definition}[section]
\newtheorem{remark}{Remark}[section]

\newtheorem{corollary}[definition]{Corollary}

\newcommand{\lp}{\left(}
\newcommand{\rp}{\right)}
\newcommand{\lc}{\left\{}
\newcommand{\rc}{\right\}}
\newcommand{\der}{\partial}

\newcommand{\bra}{\langle}
\newcommand{\ket}{\rangle}

\newcommand{\R}{\mathbb{R}}      
\newcommand{\T}{\mathbb{T}}
\newcommand{\N}{\mathbb{N}}      
\newcommand{\F}{\mathbb{F}}

\newcommand{\Flder}{\rightarrow}

\usepackage{amssymb}

\newcommand{\proa}{A^*G \mbox{$\;$}_{\tau^*} \kern-3pt\times_\alpha
G \mbox{$\;$}_\beta \kern-3pt\times_{\tau^*} A^*G}

\begin{document}

\title{Fractional damping through restricted calculus of variations}


\maketitle

\begin{center}
{\Large Fernando Jim\'enez and Sina Ober-Bl\"obaum}
\end{center}

\vspace{0.6cm}

\begin{center}
{\bf Department of Engineering Science, University of Oxford}\\
{\bf Parks Road, Oxford OXI 3PJ, UK}\\
\vspace{0.2cm}
$\,$e-mail: fernando.jimenez@eng.ox.ac.uk\\
$\,\,\,\,$e-mail: sina.ober-blobaum@eng.ox.ac.uk
\end{center}

\begin{abstract}
We deliver a novel approach towards the variational description of Lagrangian mechanical systems subject to fractional damping by establishing  a restricted Hamilton's principle. Fractional damping is a particular instance of non-local (in time) damping, which is ubiquitous in mechanical engineering applications. The restricted Hamilton's  principle relies on including fractional derivatives to the state space, the doubling of curves (which implies an extra mirror system) and the restriction of the class of varied curves. We will obtain the correct dynamics, and will show rigorously that the extra mirror dynamics is nothing but the main one in reversed time; thus, the restricted Hamilton's principle is not adding extra physics to the original system. The price to pay, on the other hand, is that the fractional damped dynamics is only a sufficient condition for the extremals of the action. In addition, we proceed  to discretise the new principle. This discretisation provides a set of numerical integrators for the continuous dynamics that we denote Fractional Variational Integrators (FVIs). The discrete dynamics is obtained upon the same ingredients, say doubling of discrete curves and restriction of the discrete variations.  We display the performance of the FVIs, which have local truncation order 1, in two examples. As other integrators with variational origin, for instance those generated by the discrete Lagrange-d'Alembert principle, they show a superior performance tracking the dissipative energy, in opposition to direct (order 1) discretisations of the dissipative equations, such as explicit and implicit Euler schemes. 
\end{abstract}

\section{Introduction}
The problem of obtaining a variational description of mechanical systems subject to external forces has been present in the literature for long time. Concretely, this article is concerned with the variational nature of the dynamical equations of a Lagrangian system subject to what we call {\it fractional damping}, namely:
\begin{equation}\label{IntEqCont}
\frac{d}{dt}\lp\frac{\der L}{\der\dot x}\rp-\frac{\der L}{\der x}=-\rho D^{\alpha}_{-}D^{\alpha}_{-}x,
\end{equation}
where $x$ is the dynamical variable, $L(x,\dot x)$ is a Lagrangian function, $\rho\in\R_+$ and $D^{\alpha}_{-}$, with $\alpha\in[0,1]$, is the retarded fractional derivative defined below \eqref{FracDer} (as usual, the dot notation represents the time derivative). The right hand side is non-local in time, and therefore the previous equation represents a particular example of non-locally damped mechanical system (we shall focus on Lagrangians given by kinetic minus potential energy), which are ubiquitous in mechanical engineering applications  (see \cite{AdWoo,Rudin,Sri} and references therein). In addition, \eqref{IntEqCont} also involves the linear damping case, since $D^{2\alpha}_{-}x= \dot x$ when $\alpha= 1/2$, where we have used \eqref{Aditive},\eqref{Alpha1}. This case is a pardigmatic example of non-Lagrangian/Hamiltonian system, i.e.~ its dynamics cannot be obtained from the Hamilton's principle \cite{AbMa} given a Lagrangian/Hamiltonian function, as proven in \cite{Bauer}. A remarkable approach towards the variational modelling of external forces, due to its phenomenological versatility, is the Lagrange-d'Alembert principle \cite{Bloch}, where the variation of the action is set equal to the work done by external forces under virtual displacements. As happens with the usual Hamilton's principle, the Lagrange-d'Alembert's can be performed  in Lagrangian and Hamiltonian fashions, both related to each other by the Legendre transformation \cite{AbMa}. It is important to note, however, that Lagrange-d'Alembert principle is not variational in the pure sense on the word, circumstance that we try to avoid in our apprach.

There are previous approaches to generate a purely variational principle by including fractional derivatives into the state space of the considered Lagrangian functions;  more concretely, we base our work on \cite{Cresson, Riewe}. These references delivered promising results, but they present some drawbacks. Say:  in  \cite{Riewe} Riewe  did not take into account the asymmetric integration by parts of the fractional derivatives \eqref{IntegrationByParts}; Cresson and collaborators, in \cite{Cresson,CressonBook}, designed the so-called {\it asymmetric embeddings} to surpass that issue, but such objects are unclear from the point of view of calculus of variations. This last approach implies the doubling of curves in the state space; i.e. it is necessary to add an extra $y$-mirror system, which is natural when treating externally forced system in a variational way, see \cite{Bateman, Galley, DaRo} and references therein.

Taking as well into account a doubled space of curves, in this work we establish a novel approach to surpass the asymmetry issue while obtaining the correct fractional damped dynamics \eqref{IntEqCont}, embodied in a particular restriction of the usual calculus of variations. Namely, we will apply the Hamilton's principle over a properly designed (in a geometric way) Lagrangian function, but we shall restrict the class of varied curves. We denote this principle by {\it restricted Hamilton's principle}. Out of this principle, we shall obtain both \eqref{IntEqCont} and the dynamics of the $y$-mirror system. However, we shall rigorously prove, both in Lagrangian and Hamiltonian settings, that the latter is nothing but the former in reversed time, which implies that the restricted Hamilton's principle is {\it not} adding extra physics to the studied system. As it is made clear below, the price to pay for applying this new principle is that the dynamical equations are not anymore necessary and sufficient conditions for the extremals of the action (as in the classical Hamilton's principle), but only sufficient.  Finally, given that the linear damping case is also included, we elucidate the connection between the restricted Hamilton's and Lagrange-d'Alembert principles.
\medskip

During the last years, the discretisation of variational principles (Hamilton's, Lagrage-d'A\-lem\-bert and others \cite{MaWe,ObJuMa}) has been of high interest in numerical integration theory, since such discretisations generate numerical integrators that approximate the systems' dynamics faithfully both from dynamical and geometrical perspectives, presenting as well a superior behaviour in the integration of the energy in the long-term. Thus, we proceed to discretise the restricted Hamilton's principle, which leads to the discrete counterpart of \eqref{IntEqDisc},  i.e.:
\begin{equation}\label{IntEqDisc}
D_1L_d(x_k,x_{k+1})+D_2L_d(x_{k-1},x_{k})=h\,\rho\,\Delta_{-}^{\alpha}\Delta_{-}^{\alpha}x_k,
\end{equation}
where $x_k$ is the discrete dynamical variable, $L_d(x_k,x_{k+1})$ the discrete Lagrangian and $\Delta_{-}^{\alpha}x_k$ is a given approximation of $D^{\alpha}_{-}x(t_k)$ for a time grid $t_k$ with time step $h$ (other approaches to discretise fractional mechanical problems can be found in \cite{Bastos, Cresson1, CreGre, CressonBook} and references therein). The elements of the discrete principle are analogous to the continuous' ones: the discrete dynamical equations are sufficient conditions for the extremals of the discrete action, and we have an extra $y$-mirror system (which again is proven to be the $x$-system in reversed discrete time). The initialisation of \eqref{IntEqDisc} as a numerical integrator of \eqref{IntEqCont} (which we will call Fractional Variational Integator, FVI) requires an initial condition which is based on the discrete Legendre transform linking the Lagrangian and Hamiltonian approaches. With that aim, we develop the restricted Hamilton's principle also in a Hamiltonian fashion (both in continuous and discrete scenarios), where, as in the usual discrete mechanics \cite{MaWe}, the momentum mathching condition will be a crucial element when constructing numerical integrators approximating the Hamiltonian version of \eqref{IntEqCont}. As in the continuous side, we elucidate the connection between the restricted Hamilton's and Lagrange-d'Alembert principles in the discrete side. Finally, we test the generated integrators with respect to well-known examples. With that aim, we use a particular benchmark approximation of the solution of inhomogeneous fractional-differential equations \cite{Vinagre}, based on a matrix discretisation of the fractional-differential operators and, afterwards, the resolution of the discrete dynamics as a matricial-algebraic equation \cite{Podlubny}.
\medskip

\hspace{-0.6cm} The paper is organised as follows: In \S\ref{Preli} we introduce the basics on fractional derivatives (since they will be used as a tool, we do not put emphasis on the mathematical specifics, and refer the interested reader to the proper literature); moreover we present the variational description of Lagrangian/Hamiltonian systems (Hamilton's principle) both in the continuous and discrete settings, as well as externally forced systems (Lagrange-d'Alembert principle). \S\ref{CRHP} is devoted to develop the continuous restricted Hamilton's principle, which is stated in its Lagrangian version in Theorem \ref{ConsFraELLag}. The relationship between the $x$-system and the $y$-system in reversed time is given in Proposition \ref{InvTime}. The preservation of a particular presymplectic form by \eqref{IntEqCont} is established in Proposition \ref{GeomPreservation}, for which the existence and uniqueness of solutions of \eqref{IntEqCont} is necessary, whose proof is developed in Appendix. The Hamiltonian version of the restricted Hamilton's principle is given in Theorem \ref{ConsFraELHam}. \S\ref{DRHP} accounts for the discrete restricted Hamilton's principle, established in Proposition \ref{FullLagProposition} and Theorem \ref{Theo:DiscEqs}. The relationship between the $x$-system and the $y$-system in reversed discrete time is given in Proposition \ref{InvDiscTime}. The connection between the discrete restricted Hamilton's principle and the discrete Lagrange-d'Alembert's is given in Corollary \ref{Corolario}, for which Lemma \ref{LemmaDisc} is essential. The definition of the discrete Legendre transform, needed for the initialisation of the FVIs, is established in Definition \ref{Def:DLT}, whereas the momentum matching condition and the  Hamiltonian version of the FVIs is given in Proposition \ref{MMPropo}. Finally, in \S\ref{Simu} we display the performance of the FVIs in two examples.

\section{Preliminaries}\label{Preli}

\subsection{Fractional derivatives}\label{FracDerPreli}

Let $\alpha\in[0,1]\subset\R$ and $f:[a,b]\subset\R\Flder\R$, $[a,b]\subset \R$, a smooth function. The fractional derivatives are defined by
\begin{small}
\begin{subequations}\label{FracDer}
\begin{align}
& _{_{RL}}D^{\alpha}_{-}f(t)=\frac{1}{\Gamma(1-\alpha)}\frac{d}{dt}\int_a^t(t-\tau)^{-\alpha}f(\tau)d\tau,\,\,\, _{_{RL}}D^{\alpha}_{+}f(t)=-\frac{1}{\Gamma(1-\alpha)}\frac{d}{dt}\int_t^b(\tau-t)^{-\alpha}f(\tau)d\tau,\label{FracDer:RL}\\
&\,\, _{_{C}}D^{\alpha}_{-}f(t)=\frac{1}{\Gamma(1-\alpha)}\int_a^t(t-\tau)^{-\alpha}\dot f(\tau)d\tau,\,\,\,\quad\,\,\, _{_{C}}D^{\alpha}_{+}f(t)=-\frac{1}{\Gamma(1-\alpha)}\int_t^b(\tau-t)^{-\alpha}\dot f(\tau)d\tau,\label{FracDer:Ca}
\end{align}
\end{subequations}
\end{small}
for $t\in$\,[a,b], and $\Gamma(z)$ the gamma function, in their Riemann-Liouville and Caputo expressions. These two kinds of fractional derivatives are related to each other, indeed it can be shown that:
\begin{equation}\label{Relation}
_{_{RL}}D^{\alpha}_{-}f(t)=\frac{-1}{\Gamma(1-\alpha)}\frac{f(a)}{(t-a)^{\alpha}}+_{_{C}}D^{\alpha}_{-}f(t),\,\,\,\,\,\,\, _{_{RL}}D^{\alpha}_{+}f(t)=\frac{1}{\Gamma(1-\alpha)}\frac{f(b)}{(b-t)^{\alpha}}+ _{_{C}}D^{\alpha}_{+}f(t).
\end{equation}
In this work, the function $f$ will represent the dynamical variable of a mechanical system. Thus, it is interesting to remark that we can always set $f(a)=0$, i.e. the system is at the origin of coordinates at initial time, and consequently both retarded Riemman-Liouville and Caputo versions are equivalent, according to the first equation in \eqref{Relation}\footnote{This is particularly apparent for $\alpha=0$. Namely: $_{_{RL}}D^{0}_{-}f(t)=f(t)$, whereas $_{_{C}}D^{0}_{-}f(t)=f(t)-f(a)$.}, for dynamical purposes. We will assume henceforth that $D_{\pm}^{\alpha}$ is determined by the Riemann-Liouville expression \eqref{FracDer:RL}, unless otherwise stated. Further relevant properties are:

\begin{subequations}
\begin{align}
\int_a^bf(t)D^{\alpha}_{\sigma}g(t)dt&=\int_a^b\lp D^{\alpha}_{-\sigma}f(t)\rp g(t)dt,\,\, \sigma=\lc-,+\rc, \label{IntegrationByParts}\\
D_{\sigma}^{\alpha}D_{\sigma}^{\beta}&=D_{\sigma}^{\alpha+\beta},\quad 0\leq\alpha,\beta\leq1,\label{Aditive}\\
D_{-}^{\alpha}= d/dt,\quad &  D_{+}^{\alpha}= -d/dt,\quad\mbox{when}\quad \alpha= 1. \label{Alpha1}
\end{align}
\end{subequations}
For the proof of these properties and more details on fractional derivatives, we refer to \cite{TheBook}.

\subsection{Continuous Lagrangian and Hamiltonian description of mechanics}\label{ContinuousFramework}

In this subsection we shall consider the configuration space of the studied systems as a finite dimensional smooth manifold $Q$. Moreover, $TQ$ and $T^*Q$ will denote its tangent and cotangent bundles, locally represented by coordinates $(q,\dot q)$ and $(q,p)$, respectively. For more details on the geometric formulation of mechanics we refer to \cite{AbMa}.

\subsubsection{Conservative systems}

Given a Lagrangian function $L:TQ\Flder\R$, the associated action functional in the time interval $t\in[a,b]\subset\R$ for a smooth curve $q:[a,b]\Flder Q$ is defined by $S(q)=\int_a^bL(q(t),\dot q(t))\,dt$.  Through Hamilton's principle, i.e. the true evolution of the system $q(t)$ with fixed endpoints $q(a)$ and $q(b)$ will satisfy

\begin{equation}\label{HPrinciple}
\delta\int_a^bL(q(t),\dot q(t))\,dt=0,
\end{equation}
we obtain the Euler-Lagrange equations via calculus of variations: 

\begin{equation}\label{ELeqs}
\frac{d}{dt}\lp\frac{\der L}{\der\dot q}\rp-\frac{\der L}{\der q}=0.
\end{equation}
Define the Legendre transformation: 
\begin{equation}\label{LegendreTrans}
\F L:TQ\Flder T^*Q; \quad (q,\dot q)\mapsto \left(q,p=\frac{\der L}{\der\dot q}\right).
\end{equation}
If  \eqref{LegendreTrans} is a global diffeomorphism we say that it is {\it hyperregular}, and we call the Lagrangian function {\it hyperregular}. Under the assumption of hyperregularity (which we will take throughout the article), via the Legendre transformation we can define the  Hamiltonian function $H:T^*Q\Flder\R$:
\begin{equation}\label{HamFunc}
H(q,p):=\bra p,\dot q\ket-L(q,\dot q),
\end{equation}
where $\bra\cdot,\cdot\ket:T^*Q\times TQ\Flder\R$ is the natural pairing. From the definition of the Hamiltonian function \eqref{HamFunc} it follows  that $L(q,\dot q)=\bra p,\dot q\ket-H(q,p)$. Furthermore, from \eqref{HPrinciple} we can write the stationary condition of the action functional in a Hamiltonian version, i.e.
\begin{equation}\label{HHPrin}
\delta\int_a^b\,\lc\bra p(t),\dot q(t)\ket-H(q(t),p(t))\rc\,dt=0.
\end{equation}
Again, using calculus of variations we obtain the Hamilton equations:

\begin{equation}\label{Heqs}
\dot q=\frac{\der H}{\der p},\quad\,\,\, \dot p=-\frac{\der H}{\der q}.
\end{equation}

We shall consider that the physical energy of the system is given by the Hamiltonian function \eqref{HamFunc}. It is easy to check that $dH(q,p)/dt=0$ under \eqref{Heqs}, showing that the system is conservative. Equivalently, the Lagrangian energy
\begin{equation}\label{LagEner}
E(q,\dot q):=\Big<\frac{\der L}{\der\dot q},\dot q\Big>-L(q,\dot q),
\end{equation}
is invariant under \eqref{ELeqs}, i.e. $dE(q,\dot q)/dt=0.$

\subsubsection{Forced systems} \label{ContForcedSystems}

First we model the external forces (which might include damping, dragging, etc.) through the mapping:
\begin{equation}\label{ExForces}
f_L:TQ\Flder T^*Q.
\end{equation}
The forced dynamics is provided by the Lagrange-d'Alembert principle \cite{Arn, Bloch}: the true evolution of the system $q(t)$ between fixed points $q(a)$ and $q(b)$ will satisfy
\begin{equation}\label{LdAPrinL}
\delta\int_a^bL(q,\dot q)\,dt+\int_a^b\bra f_L(q,\dot q),\delta q\ket\,dt=0, 
\end{equation}
where $\delta q\in TQ$, which provides the forced Euler-Lagrange equations:
\begin{equation}\label{ForcedEL}
\frac{d}{dt}\lp\frac{\der L}{\der\dot q}\rp-\frac{\der L}{\der q}=f_L(q,\dot q).
\end{equation}
Now, the Lagrangian energy of the system \eqref{LagEner} is not preserved by \eqref{ForcedEL}.  In particular $dE(q,\dot q)/dt=\big< f_L(q,\dot q),\dot q\big>$, showing that this kind of systems is not conservative.

The dual version of the Lagrange-d'Alembert principle \eqref{LdAPrinL} is naturally obtained through the Legendre transformation \eqref{LegendreTrans}. The dual external forces  $f_H:T^*Q\Flder T^*Q$ are defined by $f_H:=f_L\circ \lp\F L\rp^{-1}$ (we recall that we are assuming $L$ hyperregular), while the dynamics is established by
\begin{equation}\label{LdAPrinH}
\delta\int_a^b\lc\bra p,\dot q\ket-H(q,p)\rc\,dt+\int_a^b\bra f_H(q, p),\delta q\ket\,dt=0, 
\end{equation}
yielding the forced Hamilton equations: 
\begin{equation}\label{ForcedHam}
\dot q=\frac{\der H}{\der p},\quad\,\,\, \dot p=-\frac{\der H}{\der q}+f_H(q,p).
\end{equation}
Obviously, the Hamiltonian function \eqref{HamFunc} is not preserved under \eqref{ForcedHam}. In particular $dH(q,p)/dt=\big<f_H(q,p),\frac{\der H}{\der p}\big>$.

\subsection{Discrete Lagrangian and Hamiltonian description of mechanics}\label{DiscreteFramework}

\subsubsection{Conservative systems}

The construction of  the discrete version of mechanics relies on the substitution of $TQ$ by the Cartesian product $Q\times Q$ (note that these two spaces contain the same amount of information at local level) \cite{MaWe,MoVe}. The continuous curves $q(t)$ will be replaced by discrete ones, say $\gamma_d=\lc q_k\rc_{0:N}:=\lc q_0,q_1,...,q_N\rc\in Q^{N+1}$, where $N\in \N$ and the power $N+1$ indicates the Cartesian product of $N+1$ copies of $Q$. Given an increasing sequence of times $\lc t_k=a+hk\,|\,k=0,...,N\rc\subset \R$, with $h=(b-a)/N$, the points in $\gamma_d$ will be considered as an approximation of the continuous curve at time $t_k$, i.e. $q_k\simeq q(t_k).$ Defining the discrete Lagrangian $L_d:Q\times Q\Flder\R$  as an approximation of the action integral in one time step, say $L_d(q_k,q_{k+1},h)\simeq \int_{t_k}^{t_k+h}L(q(t),\dot q(t))\,dt$ (we shall omit the $h$ dependence of the discrete Lagrangian unless needed), we can establish the so called discrete action sum:
\begin{equation}\label{AcSum}
S_d(\gamma_d)=\sum_{k=0}^{N-1}L_d(q_k,q_{k+1}).
\end{equation}
Applying the Hamilton's principle over \eqref{AcSum}, i.e. considering variations of $\gamma_d$ with fixed endpoints $q_0=q(a),\, q_{N}=q(b)$ and extremizing $S_d$, we obtain the discrete Euler-Lagrange equations
\begin{equation}\label{DEL}
D_1L_d(q_k,q_{k+1})+D_2L_d(q_{k-1},q_{k})=0,\quad k=1,...,N-1,
\end{equation}
where $D_1$ and $D_2$ denote the partial derivative with respect to the first and second variables, respectively. If $L_d$ is regular, i.e.~the matrix $\big[D_{12}L_d\big]$ is invertible, the equations \eqref{DEL} define a discrete Lagrangian flow $F_{L_d}:Q\times Q\Flder Q\times Q$; $(q_{k},q_{k+1})\mapsto (q_{k+1},q_{k+2})$, which  is normally called {\it variational integrator} of the continuous dynamics provided by the Euler-Lagrange equations \eqref{ELeqs} (indistinctly, we shall call the equations \eqref{DEL} also variational integrator).  Moreover,  \eqref{DEL} are a discretisation in finite differences of \eqref{ELeqs}.

In order to establish the Hamiltonian picture we need to introduce the discrete Legendre transforms.
From $L_d$, two of them can be defined:
\[
\F^{-}L_d,\F^{+}L_d:Q\times Q\Flder T^*Q,
\]
in particular
\begin{subequations}\label{DLT}
\begin{align}
\F^{-}L_d(q_k,q_{k+1})&=(\,\,\,q_k\,\,\,,p_k^{-}=-D_1L_d(q_k,q_{k+1})),\label{DLT:a}\\
\F^{+}L_d(q_k,q_{k+1})&=(q_{k+1},p_{k+1}^{+}=D_2L_d(q_k,q_{k+1})).\label{DLT:b}
\end{align}
\end{subequations}
We observe that the {\it momentum matching} condition, i.e.
\begin{equation}\label{MomMat}
p_k^{-}=p_{k}^{+},
\end{equation}
provides the discrete Euler-Lagrange equations \eqref{DEL} according to \eqref{DLT} (based on this, we shall refer indistinctly to the discrete Legendre transform as momentum matching). Under the regularity of $L_d$, both discrete Legendre transforms are invertible and the discrete Hamiltonian flow $\tilde F_{L_d}:T^*Q\Flder T^*Q$; $(q_k,p_{k})\mapsto (q_{k+1},p_{k+1})$ can be defined by any of the following identities:
\begin{equation}\label{DiscHamFlow}
\tilde F_{L_d}=\F^{+}L_d\circ (\F^{-}L_d)^{-1}=\F^{+}L_d\circ F_{L_d}\circ (\F^{+}L_d)^{-1}=\F^{-}L_d\circ F_{L_d}\circ (\F^{-}L_d)^{-1};
\end{equation}
see \cite{MaWe} for the proof. At the Hamiltonian level, the map $\tilde F_{L_d}$ is called {\it variational integrator} of the continuous dynamics provided by  the Hamilton equations \eqref{Heqs}. Moreover, the discrete equations provided by \eqref{DLT} are a discretisation in finite differences of \eqref{Heqs}.

\begin{remark}\label{InitialFree}
Another advantage of the Hamiltonian version of variational integrators is that it provides a natural way of initialisating the numerical scheme. As showed in the previous discussion, two points $q_0,q_1$ are necessary in order to start \eqref{DEL} and establish the discrete flow $F_{L_d}$. On the other hand, a mechanical problem involves as initial data $q(a)=q_0$, $\dot q(a)=v_0$ and $p(a)=p_0$; thus an extra step, providing $q_1$, is in order. This step is naturally determined by \eqref{DLT:a}, leading to the algorithm (where the regularity of $L_d$ is assumed):

\begin{algorithm}{\rm Variational Integrator Scheme}
\label{Algorithm}
\begin{algorithmic}[1]
\State {\bf Initial data}: $N,\, h,\, q_0,\, p_0.$
 \State {\bf solve for} $q_1$ {\bf from} $p_0=-D_1L_d(q_0,q_1).$
\State {\bf Initial points:} $q_0,\,q_1.$
    \For {$k= 1: N-1$} 
    
{\bf solve for} $q_{k+1}$ {\bf from} $D_1L_d(q_k,q_{k+1})+D_2L_d(q_{k-1},q_{k})=0$
    \EndFor
    \State  {\bf Output:} $(q_2,...,q_N).$
\end{algorithmic}
  \end{algorithm}
Other definitions of $q_1$, different from Step 2,  can lead to the instability of the discrete flow $F_{L_d}$. \hfill $\diamond$
\end{remark}

A crucial feature of variational integrators is their {\it symplecticity}. If $\Omega_{T^*Q}$ is the canonical symplectic form on $T^*Q$ (which, according to Darboux theorem, can be locally written as $\Omega_{T^*Q}=dq\wedge dp$), define $\Omega_{Q\times Q}:=(\F^{-}L_d)^*\Omega_{T^*Q}=(\F^{+}L_d)^*\Omega_{T^*Q}$. Thus, the symplecticity of the variational integrators imply $F_{L_d}^*\Omega_{Q\times Q}=\Omega_{Q\times Q}$ \cite{MaWe, MoVe}, which furthermore imply that the energy cannot be conserved at the same time \cite{GeMa}. However, symplectic integrators have proven to present a stable energy behaviour even in long-term simulations \cite{SS}, behaviour that can be explained in terms of Backward Error Analysis \cite{BEA,HLW}.

\subsubsection{Forced systems} \label{DiscLdAlem}

As discrete version of the external forces \eqref{ExForces} we consider the maps:
\[
f_{L_d}^{-}, f_{L_d}^{+}:Q\times Q\Flder T^*Q
\]
such that
\[
\bra f_{L_d}^{-}(q_k,q_{k+1}),\delta q_k\ket+\bra f_{L_d}^{+}(q_k,q_{k+1}),\delta q_{k+1}\ket\simeq \int_{t_k}^{t_k+h}\bra f_L(q,\dot q),\delta q\ket\,dt.
\]
Note that the previous equation implies that $f_{L_d}^{-}(q_k,q_{k+1})\in T^*_{q_k}Q$ and $f_{L_d}^{+}(q_k,q_{k+1})\in T^*_{q_{k+1}}Q$. The discrete Lagrange-d'Alembert principle \cite{MaWe,ObJuMa} provides discrete curves between fixed $q_0,q_{N}$ satisfying the critical condition
\[
\delta\sum_{k=0}^{N-1}L_d(q_k,q_{k+1})+\sum_{k=0}^{N-1}\Big[\bra f_{L_d}^{-}(q_k,q_{k+1}),\delta q_k\ket+\bra f_{L_d}^{+}(q_k,q_{k+1}),\delta q_{k+1}\ket\Big]=0.
\]
These curves are given by the forced discrete Euler-Lagrange equations
\begin{equation}\label{ForcedDEL}
D_1L_d(q_k,q_{k+1})+D_2L_d(q_{k-1},q_{k})+f_{L_d}^{-}(q_k,q_{k+1})+f_{L_d}^{+}(q_{k-1},q_{k})=0,\,\, k=1,...,N-1;
\end{equation}
they are a discretisation in finite differences of \eqref{ForcedEL} and, under the regularity of the matrix $\big[D_{12}L_d(q_k,q_{k+1})+D_2f_{L_d}^{-}(q_k,q_{k+1})\big]$, provide a forced discrete Lagrangian map $F_{L_d}^{\tiny \mbox{f}}:Q\times Q\Flder Q\times Q$ approximating their continuous solution. In the forced case, the discrete Legendre transformation is defined by
\begin{subequations}\label{DForcedLT}
\begin{align}
\F^{-}L_d^{\tiny \mbox{f}}(q_k,q_{k+1})&=(\,\,\,q_k\,\,\,,\,\,p_k^{-}=-D_1L_d(q_k,q_{k+1})-f_{L_d}^{-}(q_k,q_{k+1})),\label{DForcedLT:a}\\
\F^{+}L_d^{\tiny \mbox{f}}(q_k,q_{k+1})&=(q_{k+1},\,\,p_{k+1}^{+}=D_2L_d(q_k,q_{k+1})+f_{L_d}^{+}(q_k,q_{k+1})).\label{DForcedLT:b}
\end{align}
\end{subequations}
The momentum matching condition \eqref{MomMat} reproduces the forced discrete Euler-Lagrange equations \eqref{ForcedDEL}. Moreover, \eqref{DForcedLT} provide a discretisation in finite differences of \eqref{ForcedHam}, whereas \eqref{DiscHamFlow} (for $F_{L_d}^{\tiny  \mbox{f}}$ and $\F^{\pm}L_d^{\tiny  \mbox{f}}$) yields an approximation $\tilde F_{L_d}^{\tiny  \mbox{f}}:T^*Q\Flder T^*Q$ of their continuous flow.

\begin{remark}
Establishing the forced variational flow $F_{L_d}^{\tiny  \mbox{f}}:Q\times Q\Flder Q\times Q$ requires as well two initial points $q_0,q_1$ as in the free case discussed in Remark \ref{InitialFree}. According to \eqref{DForcedLT:a} the algorithm for the forced variational integrator is given by:
\begin{algorithm}{\rm Forced Variational Integrator Scheme}
\label{ForcedAlgorithm}
\begin{algorithmic}[1]
\State {\bf Initial data}: $N,\, h,\, q_0,\, p_0.$
 \State {\bf solve for} $q_1$ {\bf from} $p_0=-D_1L_d(q_0,q_1)-f^{-}_{L_d}(q_0,q_1).$
\State {\bf Initial points:} $q_0,\,q_1.$
    \For {$k= 1: N-1$} 
    
\hspace{-0.6cm} {\bf solve for} $q_{k+1}$ {\bf from} $D_1L_d(q_k,q_{k+1})+D_2L_d(q_{k-1},q_{k})+f_{L_d}^{-}(q_k,q_{k+1})+f_{L_d}^{+}(q_{k-1},q_{k})=0$
    \EndFor
    \State  {\bf Output:} $(q_2,...,q_N).$
\end{algorithmic}
  \end{algorithm} \hfill $\diamond$
\end{remark}

\section{continuous restricted Hamilton's principle}\label{CRHP}

\subsection{Fractional state space}\label{FracPhSpa}

Consider a smooth curve $\gamma:[a,b]\subset\R\Flder\R^d$ for $d\in\N$. The local representation of the curve is given by $\gamma(t)=(x^1(t),...,x^d(t))$, $t\in[a,b]$. For the set of all smooth curves $C^{\infty}([a,b],\R^d)$, let us define the {\it fractional tangent vector} of the curve $\gamma$ by means of the following mapping
\[
\begin{split}
&X_{\sigma}^{\alpha}: C^{\infty}([a,b],\R^d) \Flder\quad\R^d,\\
&\quad\quad\quad\quad\quad\,\,\,\gamma\quad\quad\,\mapsto\,\,\,\, X_{\sigma}^{\alpha}\gamma,
\end{split}
\]
where the fractional tangent vector is defined by
\[
X_{\sigma}^{\alpha}\gamma:=D^{\alpha}_{\sigma}\gamma(t)=(D^{\alpha}_{\sigma}x^1(t),...,D^{\alpha}_{\sigma}x^d(t)),\, t\in[a,b],
\]
and $D^{\alpha}_{\sigma}$ represents the $\alpha$-fractional derivatives \eqref{FracDer}, $\sigma=\lc-,+\rc$. In the sequel we shall omit the $t$-dependence in curves and fractional tangent vectors; furthermore, we will omit as well the coordinate superindex $i$, refering to the local expression of $X_{\sigma}^{\alpha}\gamma$ just as $D^{\alpha}_{\sigma}x$.

\begin{remark}
Observe that we are choosing $\R^d$ as configuration space instead of a $d$-dimensional smooth manifold $Q$, as in \S\ref{ContinuousFramework}. We do so because the definition of $D^{\alpha}_{\sigma}\gamma(t)$ for $\gamma(t)\subset Q$ is not unique and depends on the particular set of charts employed to cover the manifold. \hfill $\diamond$
\end{remark}

\begin{definition}\label{FracVectorSpace}
{\it We define the {\rm fractional tangent space} of the curve $\gamma$ as}
{\rm
\[
V^{\alpha}_{\sigma}\R^d=\lc X_{\sigma}^{\alpha}\gamma\,|\, \mbox{for} \,\,\gamma\in C^{\infty}([a,b],\R^d)\,,\,t\in [a,b]\rc.
\]
}
\end{definition}

\begin{proposition}\label{FracVectorSpace}
{\it $V^{\alpha}_{\sigma}\R^d$ is a vector space with {\rm dim} $=d$.}
\end{proposition}
\begin{proof} The vector structure for curves is defined pointwise, i.e. for $\gamma, \delta\in C^{\infty}([a,b],\R^d)$ and  $\lambda\in\R$
\[
(\gamma+\delta)(t)=\gamma(t)+\delta(t), \quad (\lambda\gamma)(t)=\lambda\gamma(t).
\]
Noting that $D^{\alpha}_{\sigma}$ is a linear operator according to \eqref{FracDer}, the vector structure for $V^{\alpha}_{\sigma}\R^d$ is defined by 
\[
\begin{split}
X^{\alpha}_{\sigma}\gamma+X^{\alpha}_{\sigma}\delta:=X^{\alpha}_{\sigma}(\gamma+\delta),\quad
\lambda\cdot X^{\alpha}_{\sigma}\gamma:=X^{\alpha}_{\sigma}(\lambda\gamma).
\end{split}
\]
Considering that for any curve $\gamma$ the canonical basis of $\R^d$ is a linearly independent generating system with $d$ elements for $X^{\alpha}_{\sigma}\gamma$ with coordinates $D^{\alpha}_{\sigma}x^i$, we conclude that $V^{\alpha}_{\sigma}\R^d$ is a vector space with dim $=d$.  
\end{proof}

Now, we enlarge our space of cuves in the following way: given $\gamma_x,\gamma_y\in C^{\infty}([a,b],\R^d)$ it is straightforward to define the curve $\tilde\gamma\in C^{\infty}([a,b],\R^d\times\R^d)$ by $\tilde\gamma(t):=(\gamma_x(t),\gamma_y(t))$ for $t\in[a,b]$, which locally we shall denote $\tilde\gamma=(x,y)$ (which we will use indistinctly). Furthermore, for the space of curves $C^{\infty}([a,b],\R^d\times\R^d)$ we establish the vector bundle $\T^{\alpha}\R^d$, with $V^{\alpha}_{-}\R^d\times V^{\alpha}_{+}\R^d$ (whose vector structure is straightforward after Definiton \ref{FracVectorSpace} and the Cartesian product; see \cite{AbMa} for further details) as the fiber over the base space $\R^d\times\R^d$. More particuarly, for a given curve $\tilde\gamma$, $\Omega_{\tilde\gamma}^{\alpha}\in\T^{\alpha}\R^d$ is defined by $\Omega_{\tilde\gamma}^{\alpha}:=((\gamma_x,\gamma_y),(X^{\alpha}_{-}\gamma_x,X^{\alpha}_{+}\gamma_y))$, $t\in[a,b],$ and has natural coordinates $(x,y,D^{\alpha}_{-}x,D^{\alpha}_{+}y)$ (henceforth, we shall use indistinctly $\Omega_{\tilde\gamma}^{\alpha}$ and its coordinates). The bundle projection $\tau^{\alpha}:\T^{\alpha}\R^d\Flder\R^d\times\R^d$ is given by $\tau^{\alpha}(\Omega_{\tilde\gamma}^{\alpha})=(x,y)$; furthermore, it is apparent that $\T^{\alpha}\R^d\cong \R^d\times\R^d\times\R^d\times\R^d$. On the other hand, $\T\R^d$ is defined as the vector bundle with $V^{1}_{-}\R^d\times V^{1}_{-}\R^d$ as fiber over $\R^d\times\R^d$, with elements $\Xi_{\tilde\gamma}:=((\gamma_x,\gamma_y),(X^{1}_{-}\gamma_x,X^{1}_{-}\gamma_y))=((\gamma_x,\gamma_y),(D^1_{-}\gamma_x,D^1_{-}\gamma_y))=((\gamma_x,\gamma_y),(\dot\gamma_x,\dot\gamma_y))$, $t\in[a,b]$, and local coordinates $((x,y),(\dot x,\dot y)).$ The bundle projection is given by $\tau:\T\R^d\Flder\R^d\times\R^d$, $\tau(\Xi_{\tilde\gamma})=(x,y).$ With these elements we define the {\it fractional state space}.

\begin{definition}\label{FracPhSp}
{\it Consider the vector bundles $\T^{\alpha}\R^d$ and $\T\R^d$. We define the {\rm fractional state space} as the bundle product of them over $\R^d\times\R^d$, i.e.
\[
\mathfrak{T}\R^d:=\T\R^d\otimes_{\tiny{\R^d\times\R^d}}\T^{\alpha}\R^d.
\]
Thus, $\Sigma_{(x,y)}:=(\gamma_{x},\gamma_{y},X_{-}\gamma_{x}, X_{-}\gamma_{y},X_{-}^{\alpha}\gamma_{x}, X_{+}^{\alpha}\gamma_{y})\in\mathfrak{T}\R^d$, $t\in[a,b]$, is locally described by
\begin{equation}\label{TotalCoord}
\Sigma_{(x,y)}=(x,y,\dot x,\dot y,D^{\alpha}_{-}x,D^{\alpha}_{+}y).
\end{equation}
The bundle projection $\mathcal{T}:\mathfrak{T}\R^d\Flder\R^d\times\R^d$ is defined by $\mathcal{T}(\Sigma_{(x,y)})=(x,y)$. }
\end{definition}
For more details on bundle products we refer to \cite{AbMa}. The construction of the dual bundle 
\[
\mathfrak{T}^*\R^d:=\T^*\R^d\otimes_{\R^d\times\R^d}\T^{\alpha *}\R^d,
\]
which we will denote {\it fractional phase space}, follows straightforwardly from the dual bundles $\T^{\alpha *}\R^d$ and $\T^*\R^d$. 
For $P_{(x,y)}\in \mathfrak{T}^*\R^d$, we fix local coordinates
\begin{equation}\label{CoorHam}
P_{(x,y)}=(x,y,p_x,p_y,p^{\alpha}_x,p^{\alpha}_y).
\end{equation}
The bundle projection $\mathcal{P}:\mathfrak{T}^*\R^d\Flder\R^d\times\R^d$ is locally given by $\mathcal{P}(P_{(x,y)})=(x,y)$.  It is easy to see that both $\mathfrak{T}\R^d$ and $\mathfrak{T}^*\R^d$ are locally equivalent to the Cartesian product of 6 copies of $\R^d$.

\subsection{Fractional dynamics}\label{FrDy}
In order to establish the fractional dynamics we are going to consider a subclass of curves $C^{\infty}(\tilde\gamma^{_{(a,b)}};\R^d\times\R^d)\subset C^{\infty}([a,b],\R^d\times\R^d)$, in particular those  $\tilde\gamma=(\gamma_x,\gamma_y)$ such that $\gamma_x(a)=x_a,\,\gamma_x(b)=x_b$  and  $\gamma_y(a)=y_a,\,\gamma_y(b)=y_b$ for $x_a,\,x_b,\,y_a,\,y_b\in\R^d$; i.e. those curves with fixed endpoints. Let us define de action sum $S:C^{\infty}(\tilde\gamma^{_{(a,b)}};\R^d\times\R^d)\Flder\R$ for a Lagrangian function $\mathcal{L}:\mathfrak{T}\R^d\Flder\R$ (which henceforth we shall consider $C^2$) by

\begin{equation}\label{ContAction}
S(\tilde\gamma)=\int_a^b\mathcal{L}\lp x(t),y(t),\dot x(t),\dot y(t),D^{\alpha}_{-}x(t),D^{\alpha}_{+}y(t)\rp\,dt,
\end{equation}
and a particular set of varied curves and variations, namely:

\begin{definition}\label{ConsVariations}
{\it Define the set of {\rm restricted varied curves} as $\Gamma_{(\eta,\epsilon)}(t):=\tilde\gamma(t)+\epsilon\eta(t)$, $\epsilon\in\R$, where  $\Gamma_{(\eta,\epsilon)}(t)\in C^{\infty}(\tilde\gamma^{_{(a,b)}};\R^d\times\R^d)$ and $\eta(t):=(\delta\gamma_x(t),\delta\gamma_x(t))$, $t\in[a,b]$, is defined such that $\delta\gamma_x\in C^{\infty}([a,b],\R^d)$ with $\delta\gamma_x(a)=\delta\gamma_x(b)=0$. As it is easy to see, an unrestricted variation would be defined by $(\delta\gamma_x(t),\delta\gamma_y(t))$, with $\delta\gamma_x\neq\delta\gamma_y$. We impose $\delta\gamma_x(t)=\delta\gamma_y(t)$ to the variations, which locally is expressed by $\delta x=\delta y,$ vanishing at the endpoints.}
\end{definition}
Note that the introduction of $D^{\alpha}_{+}y(t)$ in the action \eqref{ContAction} breaks its causality, since it depends on future times $[t,b]$. Thus, apparently it is not suitable for the physical description of the $x$ and $y$ systems. We will see in the following discussion, however, that their dynamics can be decoupled thanks to the restriction of the variations presented in Definition \ref{ConsVariations}. Moreover, the choice of mechanical Lagrangians and the definition of the fractional derivatives \eqref{FracDer:RL} help to surpass the causality issue. 

We have already all the ingredients to establish the restricted Hamilton's principle:

\begin{theorem}\label{ConsFraELLag}
{\it A sufficient conditions for a curve $\tilde\gamma\in C^{\infty}([a,b],\R^d\times\R^d)$, subject to the restricted variations in Definition \ref{ConsVariations}, to be an extremal of the action $S:C^{\infty}(\tilde\gamma^{_{(a,b)}};\R^d\times\R^d)\Flder\R$  \eqref{ContAction} is the so-called {\rm restricted fractional  Euler-Lagrange equations}}:
\begin{subequations}\label{ResFracEL}
\begin{align}
\frac{d}{dt}\lp\frac{\der \mathcal{L}}{\der \dot x}\rp-\frac{\der \mathcal{L}}{\der x}-D^{\alpha}_{-}\lp\frac{\der \mathcal{L}}{\der D^{\alpha}_{+} y}\rp=0,\label{FracEL:-}\\
\frac{d}{dt}\lp\frac{\der \mathcal{L}}{\der \dot y}\rp-\frac{\der \mathcal{L}}{\der y}-D^{\alpha}_{+}\lp\frac{\der \mathcal{L}}{\der D^{\alpha}_{-} x}\rp=0.\label{FracEL:+}
\end{align}
\end{subequations}
\end{theorem}
\begin{proof} 
To find the extremals of $S$ for restricted varied curves $\Gamma_{(\eta,\epsilon)}(t)$ we impose the usual critical condition, i.e. $\delta S:=\frac{d}{d\epsilon}S(\Gamma_{(\eta,\epsilon)})\big|_{\epsilon=0}=0$. Using that
\[
\frac{d}{d\epsilon}\mathcal{L}(\Sigma_{_{\Gamma_{(\eta,\epsilon)}}})\Big|_{\epsilon=0}=\frac{\der \mathcal{L}}{\der x}\delta x+\frac{\der \mathcal{L}}{\der\dot x}\delta\dot x+\frac{\der \mathcal{L}}{\der D_{-}^{\alpha}x} D_{-}^{\alpha}\delta x+\frac{\der \mathcal{L}}{\der y}\delta x+\frac{\der \mathcal{L}}{\der\dot y}\delta\dot x+\frac{\der \mathcal{L}}{\der D_{+}^{\alpha}y}D_{+}^{\alpha}\delta x,
\]
where 
\[
\Sigma_{_{\Gamma_{(\eta,\epsilon)}}}=(x+\epsilon\,\delta x,y+\epsilon\,\delta x,\dot x+\epsilon\, d(\delta x)/dt,\dot y+\epsilon\, d(\delta x)/dt, D_{-}^{\alpha}x+\epsilon\,D_{-}^{\alpha}\delta x,D_{+}^{\alpha}y+\epsilon\,D_{+}^{\alpha}\delta x),
\]
according to $\Gamma_{(\eta,\epsilon)}$ in Definition \ref{ConsVariations}, we obtain:
\begin{eqnarray}
\delta S&&=\int_a^b\Big\{\frac{\der \mathcal{L}}{\der x}\delta x+\frac{\der \mathcal{L}}{\der\dot x}\delta\dot x+\frac{\der \mathcal{L}}{\der D_{-}^{\alpha}x} D_{-}^{\alpha}\delta x+\frac{\der \mathcal{L}}{\der y}\delta x+\frac{\der \mathcal{L}}{\der\dot y}\delta\dot x+\frac{\der \mathcal{L}}{\der D_{+}^{\alpha}y}D_{+}^{\alpha}\delta x\Big\}\,dt\label{Variation}\\
&&=\int_{a}^b\Big\{\frac{\der \mathcal{L}}{\der x}-\frac{d}{dt}\lp\frac{\der \mathcal{L}}{\der \dot x}\rp+D^{\alpha}_{+}\lp\frac{\der \mathcal{L}}{\der D^{\alpha}_{-} x}\rp+\frac{\der \mathcal{L}}{\der y}-\frac{d}{dt}\lp\frac{\der \mathcal{L}}{\der \dot y}\rp+D^{\alpha}_{-}\lp\frac{\der \mathcal{L}}{\der D^{\alpha}_{+} y}\rp\Big\}\delta x\,dt\nonumber\\
&&\hspace{11cm} +\frac{\der \mathcal{L}}{\der \dot x}\delta x\Big|_a^b+\frac{\der \mathcal{L}}{\der \dot y}\delta x\Big|_a^b,\nonumber
\end{eqnarray}
where in the second equality we have integration by parts with respect to the total and fractional derivatives \eqref{IntegrationByParts}. According to the vanishing endpoint conditions, all the border terms are equal to zero, leading to
\begin{eqnarray}\label{RestVar}
\delta S=-\int_{a}^b\Big\{\Big[\frac{d}{dt}\lp\frac{\der \mathcal{L}}{\der \dot x}\rp-\frac{\der \mathcal{L}}{\der x}-D^{\alpha}_{-}\lp\frac{\der \mathcal{L}}{\der D^{\alpha}_{+}y}\rp\Big]\,\delta x+\Big[\frac{d}{dt}\lp\frac{\der \mathcal{L}}{\der \dot y}\rp-\frac{\der \mathcal{L}}{\der y}-D^{\alpha}_{+}\lp\frac{\der \mathcal{L}}{\der D^{\alpha}_{-}x}\rp\Big]\delta x\Big\}\,dt.\nonumber
\end{eqnarray}
Finally, from this last expression and considering arbitrary $\delta x$, it is easy to see that the restricted fractional Euler-Lagrange equations \eqref{FracEL} are a sufficient condition for $\delta S=0$; and the claim holds. 
\end{proof}

\begin{remark}
For unrestricted variations $\delta x\neq \delta y$, both arbitrary, and fixed endpoint conditions, the necessary and sufficient conditions for the extremals of \eqref{ContAction} are the (unrestricted) fractional Euler-Lagrange equations: 
\begin{equation}\label{FracEL}
\frac{d}{dt}\lp\frac{\der \mathcal{L}}{\der \dot x}\rp-\frac{\der \mathcal{L}}{\der x}-D^{\alpha}_{+}\lp\frac{\der \mathcal{L}}{\der D^{\alpha}_{-} x}\rp=0, \quad \frac{d}{dt}\lp\frac{\der \mathcal{L}}{\der \dot y}\rp-\frac{\der \mathcal{L}}{\der y}-D^{\alpha}_{-}\lp\frac{\der \mathcal{L}}{\der D^{\alpha}_{+} y}\rp=0.
\end{equation}
Other derivations of fractional Euler-Lagrange equations can be found in \cite{Agra}. Moreover, from the linearity of the fractional derivatives \eqref{FracDer}, it is easy to prove that the restricted Euler-Lagrange equations \eqref{ResFracEL} are invariant under linear constant change of variables, i.e. $x=\Lambda\,z$,  $y=\Lambda\,\bar z$,  with $\Lambda\in \R^{d\times d}$ a regular constant matrix. If furthermore we pick the Caputo definition of the fractional derivatives \eqref{FracDer:Ca}, the equations are invariant under affine change of variables  $x=\Lambda\,z +\zeta$,  $y=\Lambda\,\bar z+\bar\zeta$; for constant $\zeta,\bar\zeta\in\R^d$ (see \cite{JiOb} for more details). \hfill $\diamond$
\end{remark}

Define now the Lagrangian $\mathcal{L}: \mathfrak{T}\R^d\Flder\R$ by
\begin{equation}\label{PartLagrangian}
\mathcal{L}(x,y,\dot x,\dot y,D^{\alpha}_{-}x,D^{\alpha}_{+}y):=L(x,\dot x)+L(y,\dot y)- D^{\alpha}_{-}x\,\rho\, D^{\alpha}_{+}y,
\end{equation}
where $L:T\R^d\Flder\R$ is a $C^2$ function, $\rho=\,$diag$(\rho_1,\cdots,\rho_d)\in \R^{d\times d}$ and we represent the matrix product $\mu^T\,M\,\nu$, for $\mu,\nu\in\R^d$ and $M\in\R^{d\times d}$, by $\mu\,M\,\nu$. In this case, the restricted fractional Euler-Lagrange equations \eqref{ResFracEL} read:
\begin{subequations}\label{ResFracELDamped}
\begin{align}
\frac{d}{dt}\lp\frac{\der L(x,\dot x)}{\der \dot x}\rp-\frac{\der L(x,\dot x)}{\der x}+\rho\,D^{\alpha}_{-} D^{\alpha}_{-} x&=0,\label{ResFracELDamped:a}\\
\frac{d}{dt}\lp\frac{\der L(y,\dot y)}{\der \dot y}\rp-\frac{\der L(y,\dot y)}{\der y}+\rho D^{\alpha}_{+} D^{\alpha}_{+} y&=0,\label{ResFracELDamped:b}
\end{align}
\end{subequations}
which are the usual Euler-Lagrange equations for a Lagrangian system plus a fractional damping term \eqref{IntEqCont}. The following proposition provides an interpretation of the $y$-mirror system.
\begin{proposition}\label{InvTime}
If we set $y(t):=x(a+b-t)$, which implies that $\tilde\gamma\in C^{\infty}(\tilde\gamma^{_{(a,b)}};\R^d\times\R^d)$ with $\gamma_y(a)=x_b,\,\gamma_y(b)=x_a$,  and we consider an even Lagrangian function $L:T\R^d\Flder\R$ in the second variable, i.e. $L(z,-\dot z)=L(z,\dot z)$, then \eqref{ResFracELDamped:b} is \eqref{ResFracELDamped:a} in reversed time.
\end{proposition}
\begin{proof}
Defining the reverse time by $\tilde t:=a+b-t$ for $t\in[a,b]$, it is easy to see that $\dot y(t)=-x^{\prime}(\tilde t)$ by applying the chain rule, where we consider $x^{\prime}(\tilde t)=dx(\tilde t)/d\tilde t$. Then
\[
\begin{split}
\frac{\der L(y(t),\dot y(t))}{\der y}&=\frac{\der L(x(\tilde t),-x^{\prime}(\tilde t))}{\der x}=^{_1}\frac{\der L(x(\tilde t),x^{\prime}(\tilde t))}{\der x},\\
\frac{d}{dt}\lp\frac{\der L(y(t),\dot y(t))}{\der \dot y}\rp&=\frac{d}{dt}\lp\frac{\der L(x(\tilde t),-x^{\prime}(\tilde t))}{\der \dot y}\rp=^{_2}-\frac{d}{dt}\lp\frac{\der L(x(\tilde t),-x^{\prime}(\tilde t))}{\der x^{\prime}}\rp=^{_3}\frac{d}{d\tilde t}\lp\frac{\der L(x(\tilde t),x^{\prime}(\tilde t))}{\der x^{\prime}}\rp,
\end{split}
\]
where in $=^{_1}$ we have used the parity $L(x,-x^{\prime})=L(x,x^{\prime})$,  and in $=^{_2}$, $=^{_3}$, the chain rule and the redefinition of the time $\tilde t=a+b-t$ and parity again, respectively. On the other hand, using \eqref{Aditive} and according to \eqref{FracDer:RL} and $y(t)=x(a+b-t)$:
\[
\begin{split}
D^{2\alpha}_{+}y(t)=&-\frac{1}{\Gamma(1-2\alpha)}\frac{d}{dt}\int_t^b(\tau-t)^{-2\alpha}x(a+b-\tau)d\tau\\
&\,\,\quad=^{_1}\frac{1}{\Gamma(1-2\alpha)}\frac{d}{dt}\int_{a+b-t}^a((a+b-t)-\tilde\tau)^{-2\alpha}x(\tilde\tau)d\tilde\tau\\
&\quad\quad\quad\quad=^{_2}\frac{1}{\Gamma(1-2\alpha)}\frac{d}{d\tilde t}\int_{a}^{\tilde t}(\tilde t-\tilde\tau)^{-2\alpha}x(\tilde\tau)d\tilde\tau\,\,=^{_3}\,\,D^{2\alpha}_{-}x(\tilde t),
\end{split}
\]
where in $=^{_1}$ we have used the change of variables $\tilde\tau=a+b-\tau$, in $=^{_2}$ the redefinition of time and, in $=^{_3}$, the definition of $D^{\alpha}_{-}$ in \eqref{FracDer:RL}. Plugging all these elements in \eqref{ResFracELDamped:b} we obtain
\[
\frac{d}{d\tilde t}\lp\frac{\der L(x(\tilde t), x^{\prime}(\tilde t))}{\der \dot x}\rp-\frac{\der L(x(\tilde t),x^{\prime}(\tilde t))}{\der x}+\rho\,D^{\alpha}_{-} D^{\alpha}_{-} x(\tilde t)=0
\]
for $\tilde t:b\Flder a$, and the claim holds.
\end{proof}
Thus, our conclusion is that, for the (mechanical) Lagrangians of interest in this work, i.e.
\begin{equation}\label{MechLag}
L(z,\dot z):=\frac{1}{2}\dot z\,m\,\dot z-U(z),
\end{equation}
with $m=\,$diag\,$(m_1,\cdots,m_d)\in\R^{d\times d}_{+}$ and $U:\R^d\Flder\R$ a smooth function, we can consider the $y$-mirror system as the $x$-system in reversed time, as long as  $\gamma_y(t)$ is bounded by $y(a)=x_b$ and $y(b)=x_a$. Thus, doubling the space of curves as exposed in \S\ref{FracPhSpa}   allows us to surpass the issue raised by the asymetric integration by parts \eqref{IntegrationByParts} without adding extra unnecessary (physical) information to the dynamics. In addition, the non-causal terms in the action  \eqref{ContAction} due to the presence of the advanced fractional derivative become causal in reversed time since $D^{\alpha}_{+}y(t)=D^{\alpha}_{-}x(\tilde t)$, as we just proved. Finally, we observe that the restricted varied curves introduced in Definition \ref{ConsVariations} acquire a particular shape when we set $y(t)=x(a+b-t)$, as we show in Figure \ref{Curves}.

\begin{figure}[!htb]
\includegraphics[scale=0.37]{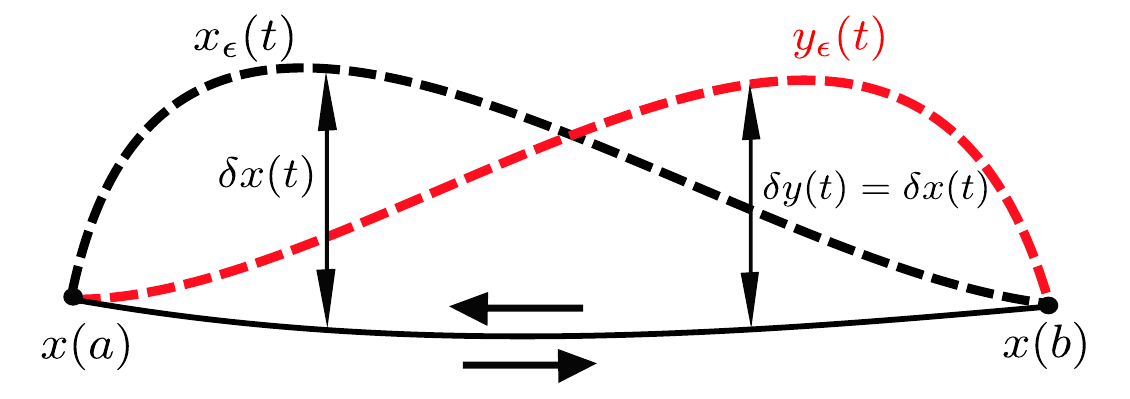}
\caption{We show the varied curves $\Gamma_{(\eta,\epsilon)}(t)=(x_{\epsilon}(t),y_{\epsilon}(t))=(x(t),y(t))+\epsilon(\delta x(t),\delta x(t))$, according to Definition \ref{ConsVariations}, when we set $y(t)=x(a+b-t)$. The horizontal arrows represent the two directions of time.}
\label{Curves}
\end{figure}

\begin{remark}
Note that the (unrestricted) fractional Euler-Lagrange equations \eqref{FracEL} for the Lagrangians \eqref{PartLagrangian} read
\[
\frac{d}{dt}\lp\frac{\der L(x,\dot x)}{\der \dot x}\rp-\frac{\der L(x,\dot x)}{\der x}+\rho\,D^{\alpha}_{+} D^{\alpha}_{+} y=0,\,\,\,\,\frac{d}{dt}\lp\frac{\der L(y,\dot y)}{\der \dot y}\rp-\frac{\der L(y,\dot y)}{\der y}+\rho\,D^{\alpha}_{-} D^{\alpha}_{-} x=0,
\]
which are a coupled system of fractional differential equations with meaningless dynamics. \hfill $\diamond$
\end{remark}

In addition,  $D^{\alpha}_{-}D^{\alpha}_{-}=D^{2\alpha}_{-}$ \eqref{Aditive}; therefore  $\alpha= 1/2$ implies $D^{2\alpha}_{-}= d/dt$ \eqref{Alpha1}. According to this, it is easy to see that, given $\alpha=1/2$, then equation \eqref{ResFracELDamped:a} is equivalent to the forced Euler-Lagrange equations \eqref{ForcedEL} when $f_L(x,\dot x)=-\rho\,\dot x$. This establishes the relationship between the restricted Hamilton's principle and the Lagrange-d'Alembert principle, expressed in Figure \ref{Diagram}.

\begin{figure}[!htb]
\includegraphics[scale=0.52]{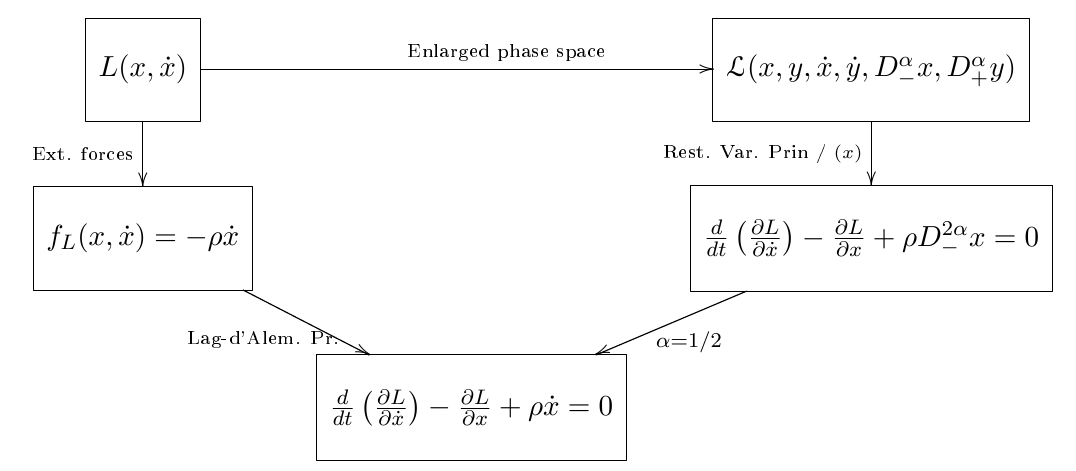}
\caption{In the diagram it is reflected that the Lagrange-d'Alembert principle for a general  Lagrangian and a particular set of external forces produce the same  forced equations as enlarging the phase space, then applying the  restricted Hamilton's principle for the action $S=\int_a^b\mathcal{L}\,dt$ \eqref{ContAction}, and then setting $\alpha=1/2$.}
\label{Diagram}
\end{figure}

Picking the mechanical Lagrangian \eqref{MechLag}, equations \eqref{ResFracELDamped} read
\begin{subequations}\label{MechDynFracDamp}
\begin{align}
m\,\ddot x+\nabla U(x)+\rho\,D^{\alpha}_{-} D^{\alpha}_{-} x&=0,\label{MechDynFracDamp:a}\\
m\,\ddot y+\nabla U(y)+\rho\,D^{\alpha}_{+} D^{\alpha}_{+} y&=0,\label{MechDynFracDamp:b}
\end{align}
\end{subequations}
which are the dynamical equations of a Lagrangian system with $2\alpha$ fractional damping, and will be considered in the examples, \S\ref{Simu}.

\subsection{Preservation of presymplectic form}

Before proving the preservation of a given presymplectic form in $\mathcal{T}\R^d$ we point out that the equations  \eqref{ResFracELDamped} define a continuous flow $F_t:\mathcal{T}\R^d\Flder \mathcal{T}\R^d$, with $F_0=$Id. The existence of this flow is ensured by  the proof of existence and uniqueness of the cited equations in Appendix. In particular, it is proven that given initial conditions $(x_0,\dot x_0)$, the existence and uniqueness of smooth $(x(t),\dot x(t))$ is ensured, ensuring furthermore the existence and uniqueness of $D_{-}^{\alpha}x(t)$ according to definition \eqref{FracDer} (equivalently for the mirror $y$-system).  We shall
ignore issues related to global versus local flows, which are easily dealt with
by restricting the domains of the flows.  Given initial conditions $X_0^{\alpha}:=(x_0,y_0,\dot x_0, \dot y_0,0,y^{\alpha}_0)\footnote{Observe that, according to \eqref{FracDer}, the only available value of $D^{\alpha}_{-}x(t_0)$ is 0.}$, the flow provides $F_t(X_0^{\alpha})=X^{\alpha}(t)$, with $X^{\alpha}(t):=(x(t),y(t),\dot x(t),\dot y(t),D^{\alpha}_{-}x(t),D^{\alpha}_{+}y(t)).$ 

Now, define the one-form and two-form on $\mathcal{T}\R^d$, with $L:T\R^d\Flder\R$ a $C^2$ function:
\begin{subequations}
\begin{align}
\Theta:=&\lp\frac{\der L(x,\dot x)}{\der\dot x^i}+\frac{\der L(y,\dot y)}{\der\dot y^i}\rp\,dx^i, \label{OneForm}\\
\Omega:=&\mbox{d}\Theta=\frac{\der^2L}{\der x^j\der\dot x^i}dx^i\wedge dx^j+\frac{\der^2L}{\der\dot x^j\der\dot x^i}dx^i\wedge d\dot x^j+\frac{\der^2L}{\der y^j\der\dot y^i}dx^i\wedge dy^j+\frac{\der^2L}{\der \dot y^j\der\dot y^i}dx^i\wedge d\dot y^j. \label{TwoForm}
\end{align}
\end{subequations}
It is straightforward to see that $\Omega$ is presymplectic on $\mathcal{T}\R^d$, since the fractional variables are absent, therefore it is degenerate, and $\mbox{d}\Omega=\mbox{d}\mbox{d}\Theta=0.$ We can establish the following preservation result.
\begin{proposition}\label{GeomPreservation}
The flow generated by the restricted fractional Euler-Lagrange equations when we set the Lagrangian function \eqref{PartLagrangian}, i.e.~ equations \eqref{ResFracELDamped}, preserves the presymplectic form $\Omega$ \eqref{TwoForm}. In other words, $F_t^*\Omega=\Omega.$
\end{proposition}

\begin{proof}
The action \eqref{ContAction}, in our fractional context, can be seen as a function $S:\mathcal{T}\R^d\Flder\R$, and its variation $\delta S=\bra \mbox{d}S\,,\,\delta X^{\alpha}_0\ket$, where we set $\delta X_0^{\alpha}:=(\delta x_0,\delta y_0, \delta\dot x_0, \delta\dot y_0,\delta x^{\alpha}_0,\delta y^{\alpha}_0)$
$\in T\mathcal{T}\R^d$ arbitrary.  According to \eqref{Variation}, when we pick the Lagrangian \eqref{PartLagrangian} and along the solutions of \eqref{ResFracELDamped}, we have that
\[
\begin{split}
\bra \mbox{d}S\,,\,\delta X^{\alpha}_0\ket=&\,\,\,\lp\frac{\der L(x,\dot x)}{\der\dot x^i}+\frac{\der L(y,\dot y)}{\der\dot y^i}\rp\,\delta x^i\Big|_{t_0}^t=^{_1}\bra\Theta(X^{\alpha}(t)),\delta X^{\alpha}(t)\ket-\bra\Theta(X^{\alpha}_0),\delta X^{\alpha}_{0}\ket\\
=^{_2}&\,\,\,\,\bra\Theta(X^{\alpha}(t)),(F_t)_{*}\delta X^{\alpha}_0\ket-\bra\Theta(X^{\alpha}_0),\delta X^{\alpha}_{0}\ket=\bra(F_t)^{*}\Theta(X^{\alpha}_0),\delta X^{\alpha}_0\ket-\bra\Theta(X^{\alpha}_0),\delta X^{\alpha}_{0}\ket\\
=&\,\,\,\bra (F_t)^{*}\Theta(X^{\alpha}_0)-\Theta(X^{\alpha}_0)\,,\,\delta X^{\alpha}_0\ket,
\end{split}
\]
where we have used arbitrary initial and final times $t_0, t$, respectively. In $=^{_1}$ we have employed the definition of $\Theta$ \eqref{OneForm}; in $=^{_2}$ we have used the flow $F_t:\mathcal{T}\R^d\Flder \mathcal{T}\R^d$ introduced above. It follows straightforwardly that
\[
\mbox{d}S=(F_t)^{*}\Theta-\Theta.
\]
Now, taking the differential in both sides, considering that the differential and the pull-back commute and the definition of $\Omega$ \eqref{TwoForm}, we obtain $(F_t)^{*}\mbox{d}\Theta-\mbox{d}\Theta=0$ and then $F_t^*\Omega=\Omega,$ from which the result follows.
\end{proof}

\subsection{The Legendre transformation}\label{LgTr}

Let us define the {\it fractional Legendre transform} $\mathcal{F}\mathcal{L}:\mathfrak{T}\R^d\Flder \mathfrak{T}^*\R^d$ as the fiber derivative for  a Lagrangian function $\mathcal{L}:\mathfrak{T}\R^d\Flder\R$, i.e.
\begin{equation}  \label{LegTrans}
\begin{split}
\mathcal{F}\mathcal{L}:\mathfrak{T}_{(x,y)}\R^d&\longrightarrow \mathfrak{T}^*_{(x,y)}\R^d\\
\Sigma_{(x,y)}&\mapsto D_{(x,y)}\mathcal{L}(\Sigma_{(x,y)}),
\end{split}
\end{equation}
where $D_{(x,y)}$ denotes the partial derivative in the fiber $\mathcal{T}^{-1}((x,y))$. Locally we have
\begin{equation}\label{FracLeg}
\mathcal{F}\mathcal{L}(\Sigma_{(x,y)})=\lp\frac{\der \mathcal{L}}{\der\dot x},\frac{\der \mathcal{L}}{\der\dot y},\frac{\der \mathcal{L}}{\der D^{\alpha}_{-}x},\frac{\der \mathcal{L}}{\der D^{\alpha}_{+}y}\rp.
\end{equation}
It is easy to check that $\mathcal{F}\mathcal{L}$ is fiber preserving \cite{AbMa}. Moreover, we will say that $\mathcal{F}\mathcal{L}$ is {\it regular} if it is a diffeomorphism, and furthermore we will call $\mathcal{L}$ {\it regular} if that is the case (which we will assume throughout the article). Hence, we define the Hamiltonian function $\mathcal{H}:\mathfrak{T}^*\R^d\Flder\R$ by
\begin{equation}\label{HamDef}
\mathcal{H}(P_{(x,y)}):=\bra \mathcal{F}\mathcal{L}(\Sigma_{(x,y)}),\Sigma_{(x,y)}\ket-\mathcal{L}(\Sigma_{(x,y)}),
\end{equation}
where the coordinates of $P_{(x,y)}:=\mathcal{F}\mathcal{L}(\Sigma_{(x,y)})$ are given in \eqref{CoorHam} and $\bra\cdot,\cdot\ket:\mathfrak{T}_{(x,y)}^*\R^d\times \mathfrak{T}_{(x,y)}\R^d\Flder\R$ denotes the natural pairing. Employing these elements, we can establish the Hamiltonian version of the restricted Hamilton's principle.

\begin{theorem}\label{ConsFraELHam}
{\it A sufficient condition for a curve $\tilde\gamma\in C^{\infty}([a,b],\R^d\times\R^d)$, subject to the restricted variations in Definition \ref{ConsVariations}, to be an extremal of the action $S:C^{\infty}(\tilde\gamma^{_{(a,b)}};\R^d\times\R^d)\Flder\R$  \eqref{ContAction} is the so-called {\rm restricted fractional  Hamilton equations}}:
\begin{subequations}\label{FracHam}
\begin{align}
\dot x&=\frac{\der \mathcal{H}}{\der p_x},\quad\quad D^{\alpha}_{-}x=\frac{\der \mathcal{H}}{\der p^{\alpha}_x}, \quad\quad \dot p_x=-\frac{\der \mathcal{H}}{\der x}+D^{\alpha}_{-}p^{\alpha}_y,\\
\dot y&=\frac{\der \mathcal{H}}{\der p_y}, \quad\quad
 D^{\alpha}_{+}y=\frac{\der \mathcal{H}}{\der p^{\alpha}_y}, \quad\quad \dot p_y=-\frac{\der \mathcal{H}}{\der y}+D^{\alpha}_{+}p^{\alpha}_x.
\end{align}
\end{subequations}
\end{theorem}
\begin{proof}
Setting the action \eqref{ContAction} in its Hamiltonian form, i.e.
\[
S(\tilde\gamma)=\int_a^b\big\{ p_x\dot x+p_y\dot y+p_x^{\alpha}D^{\alpha}_{-}x+p_y^{\alpha}D^{\alpha}_{+}y-\mathcal{H}(x,y,p_x,p_y,p^{\alpha}_x,p^{\alpha}_y)\big\}\,dt,
\]
and imposing the critical condition with restricted variations, i.e. $\delta S=\frac{d}{d\epsilon}S(\Gamma_{(\eta,\epsilon)})\big|_{\epsilon=0}=0$, after applying fractional and total integration by parts we arrive to
\[
\begin{split}
\delta S=\int_a^b&\Big\{\delta p_x\lp\dot x-\frac{\der \mathcal{H}}{\der p_x}\rp+\delta p_y\lp\dot y-\frac{\der \mathcal{H}}{\der y}\rp+\delta p_x^{\alpha}\lp D^{\alpha}_{-}x-\frac{\der \mathcal{H}}{\der p_x^{\alpha}}\rp+\delta p_y^{\alpha}\lp D^{\alpha}_{+}y-\frac{\der \mathcal{H}}{\der p_y^{\alpha}}\rp\\
&\quad\quad\quad\,\,\,+\lp-\dot p_x-\frac{\der \mathcal{H}}{\der x}+D^{\alpha}_{+}p^{\alpha}_x-\dot p_y-\frac{\der \mathcal{H}}{\der y}+D^{\alpha}_{-}p^{\alpha}_y\rp\,\delta x\Big\}dt+p_x\delta x\Big|_a^b+p_y\delta x\Big|_a^b,
\end{split}
\]
which leads to
\[
\begin{split}
\delta S=\int_a^b\Big\{&\delta p_x\lp\dot x-\frac{\der \mathcal{H}}{\der p_x}\rp+\delta p_x^{\alpha}\lp D^{\alpha}_{-}x-\frac{\der \mathcal{H}}{\der p_x^{\alpha}}\rp+\lp-\dot p_x-\frac{\der \mathcal{H}}{\der x}+D^{\alpha}_{-}p^{\alpha}_y\rp\delta x\\
+&\delta p_y\lp\dot y-\frac{\der \mathcal{H}}{\der y}\rp+\delta p_y^{\alpha}\lp D^{\alpha}_{+}y-\frac{\der \mathcal{H}}{\der p_y^{\alpha}}\rp+\lp-\dot p_y-\frac{\der \mathcal{H}}{\der y}+D^{\alpha}_{+}p^{\alpha}_x\rp\,\delta x\Big\}dt,
\end{split}
\]
since the boundary terms vanish. From this last expression is straightforward to see that \eqref{FracHam} is a sufficient condition for the extremals given arbitrary variations $\delta x,\, \delta p_x,\,\delta p_y,\, \delta p^{\alpha}_x,\,\delta p^{\alpha}_y$.
\end{proof}

\begin{remark}
For unrestricted variations $\delta x\neq \delta y$, both arbitrary, and fixed endpoint conditions, the necessary and sufficient conditions for the extremals of \eqref{ContAction} are the (unrestricted) fractional Hamilton equations: 
\[
\begin{split}
\dot x&=\frac{\der \mathcal{H}}{\der p_x},\quad\quad D^{\alpha}_{-}x=\frac{\der \mathcal{H}}{\der p^{\alpha}_x}, \quad\quad \dot p_x=-\frac{\der \mathcal{H}}{\der x}+D^{\alpha}_{+}p^{\alpha}_x,\\
\dot y&=\frac{\der \mathcal{H}}{\der p_y}, \quad\quad
 D^{\alpha}_{+}y=\frac{\der \mathcal{H}}{\der p^{\alpha}_y}, \quad\quad \dot p_y=-\frac{\der \mathcal{H}}{\der y}+D^{\alpha}_{-}p^{\alpha}_y.
\end{split}
\]
\hfill $\diamond$
\end{remark}
From \eqref{PartLagrangian} and \eqref{HamDef} we get the Hamiltonian
\begin{equation}\label{PartHamiltonian}
\mathcal{H}(x,y,p_x,p_y,p_x^{\alpha},p_y^{\alpha})=H(x,p_x)+H(y,p_y)-p_x^{\alpha}\,\rho^{-1}\,p^{\alpha}_y,
\end{equation}
where $H(z,p_z):=\bra p_z,\,\dot z\ket-L(z,\dot z)$. For the particular Hamiltonian \eqref{PartHamiltonian}, the restricted Hamilton equations \eqref{FracHam} read:

\begin{subequations}\label{FracPartiHam}
\begin{align}
\dot x&=\frac{\der H}{\der p_x},\quad\quad D^{\alpha}_{-}x=-\rho^{-1}\,p_y^{\alpha}, \quad\quad \dot p_x=-\frac{\der H}{\der x}+D^{\alpha}_{-}\,p^{\alpha}_y,\label{FracPartiHam:a}\\
\dot y&=\frac{\der H}{\der p_y}, \quad\quad
 D^{\alpha}_{+}y=-\rho^{-1}\,p_x^{\alpha}, \quad\quad \dot p_y=-\frac{\der H}{\der y}+D^{\alpha}_{+}\,p^{\alpha}_x.\label{FracPartiHam:b}
\end{align}
\end{subequations}
We see that the $x$ and $y$ dynamics are also coupled in principle. On the other hand, in the next result we show that the relationship between $x$ and $y$ systems is consistent with the Lagrangian case under inversion of time.
\begin{proposition}
If we set $y(t):=x(a+b-t)$, $p_y(t):=-p_x(a+b-t)$ and $p_y^{\alpha}(t):=p_x(a+b-t)$,  and we consider an even Hamiltonian function in the second variable, i.e. $H(z,-p_z)=H(z,p_z)$, then \eqref{FracPartiHam:b} is \eqref{FracPartiHam:a} in reversed time.
\end{proposition}
\begin{proof} 
Given that $\dot y(t)=-x^{\prime}(\tilde t)$  and
\[
\frac{\der H(y(t),p_y(t))}{\der p_y}=^{_{1}}\frac{\der H(x(\tilde t),-p_x(\tilde t))}{\der p_y}=^{_{2}}-\frac{\der H(x(\tilde t),-p_x(\tilde t))}{\der p_x}=^{_{3}}-\frac{\der H(x(\tilde t),p_x(\tilde t))}{\der p_x},
\]
where in $=^{_{1}}$ we use the hypotheses, in $=^{_{2}}$ the chain rule and in $=^{_{3}}$ the parity of the Hamiltonian;  we see that	 
\[
\dot y(t)=\frac{\der H(y(t),p_y(t))}{\der p_y} \Rightarrow -x^{\prime}(\tilde t)=-\frac{\der H(x(\tilde t),p_x(\tilde t))}{\der p_x}\Rightarrow x^{\prime}(\tilde t)=\frac{\der H(x(\tilde t),p_x(\tilde t))}{\der p_x}.
\]
From the proof of Proposition \ref{InvTime}, $D^{\beta}_{+}y(t)=D^{\beta}_{-}x(\tilde t)$, and therefore it is easy to see that
\[
D^{\alpha}_{+}y(t)=-\rho^{-1}\,p^{\alpha}_x(t)\Rightarrow D^{\alpha}_{-}x(\tilde t)=-\rho^{-1}\,p^{\alpha}_y(\tilde t).
\]
Finally, considering 
\[
\frac{\der H(y(t),p_y(t))}{\der y}=\frac{\der H(x(\tilde t),-p_x(\tilde t))}{\der y}=\frac{\der H(x(\tilde t),p_x(\tilde t))}{\der x},
\]
where again we use the chain rule and the parity of the Hamiltonain, and that $\dot p_y(t)=p_x^{\prime}(\tilde t)$ according to the hypotheses, we get
\[
\dot p_y(t)=-\frac{\der H(y(t),p_y(t))}{\der y}+D^{\alpha}_{+}\,p_x(t)\Rightarrow p_x^{\prime}(\tilde t)=-\frac{\der H(x(\tilde t),p_x(\tilde t))}{\der x}+D^{\alpha}_{-}p_y^{\alpha}(\tilde t),
\]
and the claim holds.
\end{proof}
The mechanical Hamiltonian
\begin{equation}\label{MechHamilto}
H(z,p_z)=\frac{1}{2}p_z\,m^{-1}\,p_z+U(z)
\end{equation}
presents even parity in the second variable, and therefore the previous result applies. The restricted fractional Hamilton equations \eqref{FracPartiHam} in this case, after replacing the fractional relationship in the dynamical equation for $p_x$ and $p_y$, read
\begin{subequations}\label{FracHamEqs:Mech}
\begin{align}
\dot x&=m^{-1}\,p_x, \quad\quad \dot p_x=-\nabla U(x)-\rho\,\,D^{\alpha}_{-}D^{\alpha}_{-}\,x,\\
\dot y&=m^{-1}\,p_y, \quad\quad \dot p_y=-\nabla U(y)-\rho\,\,D^{\alpha}_{+}D^{\alpha}_{+}\,y,
\end{align}
\end{subequations}
where we recognize the usual Hamilton equations for mechanical systems plus a fractional damping term. We observe that, after getting rid of the pure fractional equation, the dynamics of $x$ and $y$ are again decoupled. Furthermore, when $\alpha = 1/2$ it is easy to see that the $x$-equations (analogous arguments can be applied to the $y$-equations) are equivalent to the forced Hamilton equations \eqref{ForcedHam} when $f_H(x,p_x)=-\rho\,m^{-1}p_x$. Indeed, when $\alpha=1/2$:
\[
\dot x=^{_{1}}m^{-1}p_x, \quad\quad \dot p_x=-\nabla U(x)-\rho\,D^{_{1/2}}_{-}D^{_{1/2}}_{-}\,x=^{_{2}}-\nabla U(x)-\rho\,\dot x=-\nabla U(x)-\rho\,m^{-1}p_x,
\]
 which are the usual Hamilton equations for mechanical systems with linear damping. 

\begin{remark}
We observe in the right hand side of $=^{_{2}}$ that the damping force is not actually defined in $T^*\R^d$, but in $T\R^d$, i.e. $f_L(x,\dot x)=-\rho\,\dot x$, and that we can relate it to a pure ``cotangent'' force thanks to $=^{_{1}}$, say $f_H(x,p_x)=-\rho\,m^{-1}p_x$. This is a specific phenomenon of our approach, and differs from the usual description of forced systems \S\ref{ContForcedSystems}. However, we observe as well that both dynamics, i.e.
\begin{eqnarray}
&\dot x&=m^{-1}p_x,\quad\quad \dot p_x=-\nabla U(x)-\rho\,\dot x,\nonumber\\
&\dot x&=m^{-1}p_x,\quad\quad \dot p_x=-\nabla U(x)-\rho\,m^{-1}p_x, \label{HamLinearDamp}
\end{eqnarray}
define the same subspace in $TT^*\R^d$, whose natural coordinates are $(x,p_x,\dot x,\dot p_x).$ \hfill $\diamond$
\end{remark}

\begin{remark}
Define the  one-form on $\mathfrak{T}^*\R^d$ given by $\Theta_H:=(\mathcal{F}\mathcal{L})^*\Theta$, with $\Theta$ defined in \eqref{OneForm} and  $\mathcal{F}\mathcal{L}$ in \eqref{FracLeg}. Using the Lagrangian \eqref{PartLagrangian} and the coordinates \eqref{CoorHam}, it is easy to see that $\Theta_H=(p_x+p_y)\,dx$. Using analogous arguments than in Proposition \ref{GeomPreservation}, it is easy to see that the flow generated by equations \eqref{FracPartiHam} preserves the presymplectic form $\Omega_H:=\mbox{d}\Theta_H=dx\wedge dp_x+dx\wedge dp_y$.  \hfill $\diamond$ 
\end{remark}

\section{Discrete restricted Hamilton's principle}\label{DRHP}

\subsection{Discrete Lagrangian dynamics}

Let us consider the increasing sequence of times $\{ t_k=a+hk\,|\,k=0,...,N\}\subset\R$ where $h$ is the fixed time step  determined by $h=(b-a)/N$. Define a discrete curve as a collection of points in $\R^d$ i.e. $\gamma_d:=\lc x_0,x_1,...,x_{N-1},x_N\rc=\lc x_k\rc_{0:N}\in\R^{(N+1)d}$ (here $\R^{(N+1)d}$ denotes $(\R^d)^{(N+1)}:=\R^d\times\cdots\times\R^d$, the Cartesian product of $(N+1)$ copies of $\R^d$). As  usual, we will consider these points as an approximation of the continuous curve at time $t_k$, namely $x_k\simeq x(t_k).$ Given $\lc z_k\rc_{0:N}$ (later on we shall particularise in $\lc x_k\rc_{0:N}$ and $\lc y_k\rc_{0:N}$) define the following sequences:
\begin{subequations}\label{AllSequencs}
\begin{align}
&\lc S^{\kappa}\,z_k\rc_{\tiny 0:N-1},\,\,\,\quad\, S^{\kappa}\,z_k:=\kappa\,z_k+(1-\kappa)\,z_{k+1},\,\,\quad\kappa\in[0,1]\subset\R,\label{KappaRule}\\
&\lc \Delta_{-}^{\alpha}z_k\rc_{\tiny 0:N},\,\,\,\,\,\,\,\,\quad\Delta_{-}^{\alpha}z_k:=\frac{1}{h^{\alpha}}\sum_{n=0}^k\alpha_nz_{k-n},\label{DiscFracDerDef:1}\\
&\lc \Delta_{+}^{\alpha}z_k\rc_{\tiny 0:N},\,\,\,\,\,\,\,\,\quad\Delta_{+}^{\alpha}z_k:=\frac{1}{h^{\alpha}}\sum_{n=0}^{N-k}\alpha_nz_{k+n},\label{DiscFracDerDef:2}
\end{align}
\end{subequations}
where

\begin{equation}\label{AlphaDef}
\alpha_n:=\frac{-\alpha\,(1-\alpha)\,(2-\alpha)\cdot\cdot\cdot(n-1-\alpha)}{n!};\quad\alpha_0:=1.
\end{equation}
The discrete series $\Delta_{-}^{\alpha}x_k$ (resp. $\Delta_{+}^{\alpha}y_k$) are an approximation of $D^{\alpha}_{-} x(t_k)$ (resp. $D^{\alpha}_{+} y(t_k)$).  For more details on the approximation of fractional derivatives we refer to (\cite{CressonBook}, Chapter 5).
\begin{remark}\label{SumEndpoint}
We observe that  \eqref{DiscFracDerDef:1}, \eqref{DiscFracDerDef:2} are well-defined for $k=0$ and $k=N$. Namely, straightforward computations lead to $\Delta_{-}^{\alpha}x_0=\alpha_0x_0/h^{\alpha}$ and $\Delta_{+}^{\alpha}y_N=\alpha_0y_N/h^{\alpha}$. \hfill  $\diamond$
\end{remark}
Given two sequences $\lc F_k\rc_{\tiny 0:N}, \lc G_k\rc_{\tiny 0:N}$, the discrete fractional derivatives \eqref{DiscFracDerDef:1}, \eqref{DiscFracDerDef:2}, obey the following discrete integration by parts relationships:
\begin{lemma}
Consider $\lc F_k\rc_{\tiny 0:N}, \lc G_k\rc_{\tiny 0:N}$, with $F_0=F_N=G_0=G_N=0$. Then
\begin{subequations}
\begin{align}
\sum_{k=0}^{N-1}(\Delta_{+}^{\alpha}G_k)F_k&=\sum_{k=1}^NG_k(\Delta_{-}^{\alpha}F_k). \label{DIParts:b}\\
\sum_{k=0}^{N-1}G_{k+1}(\Delta_{-}^{\alpha}F_{k+1})&=\sum_{k=1}^{N-1}(\Delta_{+}^{\alpha}G_{k})F_{k}. \label{DIParts:c}
\end{align}
\end{subequations}
\end{lemma}
\begin{proof}

\eqref{DIParts:b} (see \cite{Cresson1} for more details):
\[
\begin{split}
\sum_{k=1}^NG_k(\Delta_{-}^{\alpha}F_k)&=^{_1}\frac{1}{h^{\alpha}}\sum_{k=1}^N\sum_{n=0}^k\alpha_nG_kF_{k-n}=^{_2}\frac{1}{h^{\alpha}}\sum_{k=0}^N\sum_{n=0}^k\alpha_nG_kF_{k-n}=^{_3}\frac{1}{h^{\alpha}}\sum_{n=0}^N\sum_{k=n}^N\alpha_nG_kF_{k-n}\\
&=^{_4}\frac{1}{h^{\alpha}}\sum_{n=0}^N\sum_{k=0}^{N-n}\alpha_nG_{k+n}F_{k}=^{_5}\frac{1}{h^{\alpha}}\sum_{k=0}^N\sum_{n=0}^{N-k}\alpha_nG_{k+n}F_{k}=^{_6}\frac{1}{h^{\alpha}}\sum_{k=0}^{N-1}\sum_{n=0}^{N-k}\alpha_nG_{k+n}F_{k}\\
=^{_7}\sum_{k=0}^{N-1}(\Delta_{+}^{\alpha}G_k)F_k&.
\end{split}
\]
In $=^{_1}$ the definition \eqref{DiscFracDerDef:1} is used. In $=^{_2}$  $F_0=0$ is taken into account. To prove $=^{_3}$  is enough to notice that, for a fixed $j=0,...,N,$  the elements  $a_{i}:=\alpha_iG_jF_{j-i}$, $i=0,...,j$, on the left hand side, disposed in columns, form an upper diagonal $(N+1)\times (N+1)$, whereas the same elements on the right hand side, for $j=0,...,N$ and $i=j,...,N$, account for the transposed matrix; therefore their total sums are equal. In $=^{_4}$ the sum index is rearranged. In $=^{_5}$ equivalent arguments to $=^{_3}$ can be used. In $=^{_6}$ $F_N=0$ is taken into account. Finally, in $=^{_7}$ the definition \eqref{DiscFracDerDef:2} is used.

\eqref{DIParts:c}:

\[
\sum_{k=0}^{N-1}G_{k+1}(\Delta_{-}^{\alpha}F_{k+1})=^{_{1}}\sum_{k=1}^{N}G_{k}(\Delta_{-}^{\alpha}F_{k})=^{_{2}}\sum_{k=0}^{N-1}(\Delta_{-}^{\alpha}G_{k})F_{k}=^{_{3}}\sum_{k=1}^{N-1}(\Delta_{+}^{\alpha}G_{k})F_{k}.
\]
In $=^{_{1}}$ we have arranged the sum index. In $=^{_{2}}$ we have used \eqref{DIParts:b} and, finally, in $=^{_{3}}$ we have used $F_0=0.$
 \end{proof}

\begin{remark}\label{MixedIndices}
By simple inspection it is easy to check that the discrete asymetric integration by parts does not hold for {\it out of phase} indices, for instance $\sum_{k=0}^{N-1}(\Delta^{\alpha}_{+}G_k)F_{k+1}\neq \sum_{k=1}^{N-1}G_k(\Delta^{\alpha}_{-}F_{k+1})$ and similar cases. \hfill $\diamond$
\end{remark}

Next, according to \S\ref{DiscreteFramework}, we shall consider the discrete Lagrangian as an approximation of the continuous action \eqref{ContAction}, i.e.
\[
\mathcal{L}_d\simeq \int_{t_k}^{t_k+h}\mathcal{L}\lp x(t),y(t),\dot x(t),\dot y(t),D^{\alpha}_{-}x(t),D^{\alpha}_{+}y(t)\rp\,dt,
\]
for $a+h<t_k< b-h$. As mentioned above, we shall use \eqref{DiscFracDerDef:1},\eqref{DiscFracDerDef:2} as discrete counterparts for the fractional derivatives. We see that these discrete series imply the whole discrete past for $x$ and the whole discrete future for $y$ at time $t_k$.  Thus, it is clear that the approximation of the action in the discrete time interval $[k,k+1]$ would depend on $x_{(0,k+1)}:=(x_0,...,x_{k+1})\in \R^{d(k+2)}$ and $y_{(k,N)}:=(y_k,...,y_{N})\in \R^{d(N+1-k)}$, accounting for the discrete version of the non-locality of the fractional derivatives. Moreover, it is explicit that the approximation of the action shall depend as well on the interval $[k,k+1]$ and therefore on $k$. According to this, we establish the function
\begin{equation}\label{DiscLag}
\begin{split}
\mathcal{L}_d^k:&\,\,\,\R^{d(k+2)}_X\times \R^{d(N+1-k)}_Y\quad\longrightarrow\quad\R\\
&\,\,\,\,\,(x_{(0,k+1)},y_{(k,N)})\quad\quad\quad\mapsto \mathcal{L}_d^k(x_{(0,k+1)},y_{(k,N)}),
\end{split}
\end{equation}
where we remark the correspondence of the first $k+2$ entries with the discrete $x$-path, and the last $N+1-k$ with the $y$-path.  If we define the sets of discrete curves as

\[
\begin{split}
C_d^x&=\lc \gamma_d^x=\lc x_k\rc_{\tiny 0:N}\in\R^{d(N+1)}\,|\,x_0=x_a,\,\,x_N=x_b\rc,\\
C_d^y&=\lc \gamma_d^y=\lc y_k\rc_{\tiny 0:N}\in\R^{d(N+1)}\,|\,\,y_0=y_a,\,\,\,y_N=y_b\rc,
\end{split}
\]
then the discrete action sum is defined naturally by $S_d:C_d^x\times C_d^y\Flder\R$:
\begin{equation}\label{DiscAct}
S_d(\tilde\gamma_d):=\sum_{k=0}^{N-1}\mathcal{L}_d^k(x_{(0,k+1)},y_{(k,N)}),
\end{equation}
where $\tilde\gamma_d:=(\gamma_d^x,\gamma_d^y)\in C_d^x\times C_d^y$. Next, we introduce the discrete restricted variations.

\begin{definition}\label{DiscreteVariations}
{\it Given a discrete curve $\tilde\gamma_d=(\gamma_d^x,\gamma_d^y)\in C_d^x\times C_d^y$, we define the set of {\rm  varied discrete curves} by
\begin{equation}\label{DiscVar}
\Gamma^x_{\epsilon}:=\gamma_d^x+\epsilon\,\delta\gamma^x_d,\quad
\Gamma^y_{\epsilon}:=\gamma_d^y+\epsilon\,\delta\gamma^y_d.
\end{equation}
where $\delta\gamma_d^x:=\lc \delta x_k\rc_{0:N}$, $\delta\gamma_d^y:=\lc \delta y_k\rc_{0:N}$ are the {\rm discrete variations}, defined such that 
\begin{equation}\label{DiscEndpoints}
\delta x_0=\delta x_N=0,\quad\,\delta y_0=\delta y_N=0.
\end{equation}
We define the set of {\rm restricted  varied discrete curves}, by
\begin{equation}\label{DiscRestVar}
\Gamma^x_{(\epsilon,\eta)}:=\gamma_d^x+\epsilon\,\eta_d,\quad
\Gamma^y_{(\epsilon,\eta)}:=\gamma_d^y+\epsilon\,\eta_d,
\end{equation}
where we establish $\eta_d=\delta\gamma_d^x=\delta\gamma_d^y$. In other words, we set $\delta x_k=\delta y_k$ for $k=1,...,N-1$.}
\end{definition}
In the following proposition we establish the necessary and sufficient condition for the extremals of \eqref{DiscAct} under restricted variations (discrete restricted Hamilton's principle). 

\begin{proposition}\label{FullLagProposition}
Given the discrete action \eqref{DiscAct} subject to restricted variations \eqref{DiscRestVar} and endpoint conditions \eqref{DiscEndpoints}, the necessary and sufficient conditions for the extremals are
\begin{equation}\label{DELEqs:NS} 
\sum_{i=k-1}^{N-1}D_{x_k}\mathcal{L}_d^i(x_{(0,i+1)},y_{(i,N)}) +\sum_{i=0}^{k}D_{y_k}\mathcal{L}_d^i(x_{(0,i+1)},y_{(i,N)})=0,
\end{equation}
for $k=1,...,N-1$, where $D_{x_k}:=\der/\der x_k$ and $D_{y_k}:=\der/\der y_k$.
\end{proposition}

\begin{proof}
The condition for the extremals is $\frac{d}{d\epsilon}S_d((\Gamma^x_{(\epsilon,\eta)},\Gamma^y_{(\epsilon,\eta)}))\Big|_{\epsilon=0}=0$, which, as it is easy to see, is equivalent to 
\[
\begin{split}
\delta \sum_{k=0}^{N-1}&\mathcal{L}_d^k(x_{(0,k+1)},y_{(k,N)})\\
&=\delta \mathcal{L}_d^0(x_{(0,1)},y_{(1,N)}) + \delta \mathcal{L}_d^1(x_{(0,2)},y_{(1,N)})+\cdots\\
&\hspace{4cm} \cdots+\delta \mathcal{L}_d^{N-2}(x_{(0,N-1)},y_{(N-2,N)})+\delta \mathcal{L}_d^{N-1}(x_{(0,N)},y_{(N-1,N)})=0,
\end{split}
\]
given that $\delta x_k=\delta y_k$. Considering the variations with respect to $x$ and $y$ separately, we obtain
\[
\begin{split}
&\delta_x\sum_{k=0}^{N-1}\mathcal{L}_d^k(x_{(0,k+1)},y_{(k,N)})\\
=\,&D_{_{x_0}}\mathcal{L}_d^0\,\delta x_0+D_{_{x_1}}\mathcal{L}_d^0\,\delta x_1\\
+&D_{_{x_0}}\mathcal{L}_d^1\,\delta x_0+D_{_{x_1}}\mathcal{L}_d^1\,\delta x_1+D_{_{x_2}}\mathcal{L}_d^1\,\delta x_2\\
&\hspace{0.6cm} \vdots\hspace{2.4cm} \vdots\hspace{2.4cm} \vdots\\
+&D_{_{x_0}}\mathcal{L}_d^{_{N-1}}\,\delta x_0+D_{_{x_1}}\mathcal{L}_d^{_{N-1}}\,\delta x_1+D_{_{x_2}}\mathcal{L}_d^{_{N-1}}\,\delta x_2+\cdots+D_{_{x_{N-1}}}\mathcal{L}_d^{_{N-1}}\,\delta x_{N-1}+D_{_{x_N}}\mathcal{L}_d^{_{N-1}}\,\delta x_N,
\end{split}
\]
where $D_{x_i}$ represents the partial derivative with respect to $x_i$, and
\[
\begin{split}
&\delta_y\sum_{k=0}^{N-1}\mathcal{L}_d^k(x_{(0,k+1)},y_{(k,N)})\\
=\,&D_{_{y_0}}\mathcal{L}_d^{0}\,\delta x_0+D_{_{y_1}}\mathcal{L}_d^{0}\,\delta x_1+D_{_{y_2}}\mathcal{L}_d^{0}\,\delta x_2+\cdots+D_{_{y_{N-1}}}\mathcal{L}_d^{0}\,\delta x_{N-1}+D_{_{y_N}}\mathcal{L}_d^{0}\,\delta x_N\\
&\hspace{1.7cm} +D_{_{y_1}}\mathcal{L}_d^1\,\delta x_1+D_{_{y_2}}\mathcal{L}_d^1\,\delta x_2+\cdots+D_{_{y_{N-1}}}\mathcal{L}_d^{1}\,\delta x_{N-1}+D_{_{y_N}}\mathcal{L}_d^{1}\,\delta x_N\\
&\hspace{9cm} \vdots\hspace{2.2cm}\vdots\\
&\hspace{6.2cm} +D_{_{y_{N-1}}}\mathcal{L}_d^{_{N-1}}\,\delta x_{N-1}+D_{_{y_N}}\mathcal{L}_d^{_{N-1}}\,\delta x_N.
\end{split}
\]
Adding up the columns, we arrive to
\[
\begin{split}
&\delta \sum_{k=0}^{N-1}\mathcal{L}_d^k(x_{(0,k+1)},y_{(k,N)})=\Big(\sum_{i=0}^{N-1}D_{_{x_0}}\mathcal{L}_d^i+D_{_{y_0}}\mathcal{L}_d^0\Big)\,\delta x_0\\
&\hspace{3cm} +\sum_{k=1}^{N-1}\Big(\sum_{i=k-1}^{N-1}D_{x_k}\mathcal{L}_d^i +\sum_{i=0}^{k}D_{y_k}\mathcal{L}_d^i\Big)\,\delta x_k
+\Big( D_{_{x_N}}\mathcal{L}_d^{N-1}+\sum_{i=0}^{k}D_{y_k}\mathcal{L}_d^i\Big)\,\delta x_N.
\end{split}
\]
Equating the last expression to 0, considering $\delta x_0=\delta x_N=0$ and arbitrary variations $\delta x_k$, we arrive directly to \eqref{DELEqs:NS}.
\end{proof}

\begin{remark}\label{AllVariables}
Observe that the discrete Lagrangian problem established by \eqref{DiscLag} and \eqref{DiscAct} is of the higher-order type, i.e. the discrete Lagrangian depends on multiple copies (more than two) of the configuration manifold (see \cite{HO:1,HO:2} for more details). As striking difference with respect to the kind of problem in these references, we note that the number of copies of $\R^d_X$ and $\R^d_Y$ the discrete Lagrangian \eqref{DiscLag} depends on is not fixed and is determined by $k$. This circumstance 
prevents in general the definition of a discrete flow in the sense expressed in \S\ref{DiscreteFramework}, where  $F_{L_d}:Q\times Q\Flder Q\times Q$ can be obtained from \eqref{DEL} under regularity conditions.  This can be seen as well noticing that all variables are present at the same time in \eqref{DELEqs:NS} for any $k$, which makes mandatory to solve them simultaneously in order to obtain the sequences $\lc x_k\rc_{0:N}$, $\lc y_k\rc_{0:N}$. However, we will see below that particular choices of $\mathcal{L}_d^k$ lead to actual discrete flows.
\hfill $\diamond$
\end{remark}

\begin{remark}\label{RemarkLagrangianMixing}
For unrestricted variations \eqref{DiscVar}, i.e. $\delta x_k\neq\delta y_k$ and both arbitrary, and endpoint conditions \eqref{DiscEndpoints}, the necessary and sufficient conditions for the extremals of \eqref{DiscAct} read
\[
\sum_{i=k-1}^{N-1}D_{x_k}\mathcal{L}_d^i(x_{(0,i+1)},y_{(i,N)})=0,\,\,\,\quad\sum_{i=0}^{k}D_{y_k}\mathcal{L}_d^i(x_{(0,i+1)},y_{(i,N)})=0;\,\,\, k=1,...,N-1. \quad\quad\quad \diamond
\]
\end{remark}

Next, let us pick the discrete Lagrangian
\begin{equation}\label{DiscLag:Dyn}
\mathcal{L}_d^k(x_{(0,k+1)},y_{(k,N)}):=L_d(x_k,x_{k+1})+L_d(y_k,y_{k+1})-h\,\Delta_{-}^{\alpha}x_k\,\rho\,\Delta_{+}^{\alpha}y_k,
\end{equation}
where $L_d:\R^d\times\R^d\Flder\R$ is a particular discretisation of the continuous action integral defined for the Lagrangian $L:T\R^d\Flder\R$ in \eqref{PartLagrangian}, the discrete fractional derivatives are defined in \eqref{DiscFracDerDef:1}, \eqref{DiscFracDerDef:2} and $\rho\in\R^{d\times d}$. Furthermore, we set $\mathcal{L}_d^N:=0$ since $L_d$ is not defined in such a case. Naturally, this is the discrete counterpart of the continuous Lagrangian \eqref{PartLagrangian}.

\begin{theorem}\label{Theo:DiscEqs}
Given the discrete Lagrangian \eqref{DiscLag:Dyn}, restricted variations \eqref{DiscRestVar} and endpoint conditions \eqref{DiscEndpoints}, a sufficient condition for the extremals of \eqref{DiscAct} is 
\begin{subequations}\label{RestDEL:Dyn}
\begin{align}
&D_1L_d(x_k,x_{k+1})+D_2L_d(x_{k-1},x_{k})-h\,\rho\,\Delta_{-}^{\alpha}\Delta_{-}^{\alpha}x_k=0,\label{RestDEL:Dyn:a}\\
&D_1L_d(y_k,y_{k+1})\,+\,D_2L_d(y_{k-1},y_{k})-h\,\,\rho\,\Delta_{+}^{\alpha}\Delta_{+}^{\alpha}y_k=0,\label{RestDEL:Dyn:b}
\end{align}
\end{subequations}
for $k=1,...,N-1$, and $D_i$ denotes the partial derivative with respect to the $i$-th variable.
\end{theorem}

\begin{proof}
To prove the claim it is enough to show that equations \eqref{RestDEL:Dyn} are a sufficient condition of \eqref{DELEqs:NS} to be satisfied.

First,  we see that
\begin{subequations}\label{Sums}
\begin{align}
\frac{\der}{\der x_k}\sum_{i=k-1}^{N-1}\mathcal{L}_d^i&=\frac{\der}{\der x_k}\Big(-\sum_{i=0}^{k-2}\mathcal{L}_d^i+\sum_{i=0}^{N-1}\mathcal{L}_d^i\Big)=\frac{\der}{\der x_k}\sum_{i=0}^{N-1}\mathcal{L}_d^i=\frac{\der}{\der x_k}\sum_{i=0}^{N}\mathcal{L}_d^i,\label{Sums:a}\\
\frac{\der}{\der y_k}\sum_{i=0}^{k}\mathcal{L}_d^i&=\frac{\der}{\der y_k}\Big(-\sum_{i=k+1}^{N-1}\mathcal{L}_d^i+\sum_{i=0}^{N-1}\mathcal{L}_d^i\Big)=\frac{\der}{\der y_k}\sum_{i=0}^{N-1}\mathcal{L}_d^i.\label{Sums:b}
\end{align}
\end{subequations}
In \eqref{Sums:a}  we use that  $\sum_{i=0}^{k-2}\der\mathcal{L}_d^i/\der x_k=0$ because there is no dependence of $\mathcal{L}_d^i$ on $x_k$ in the range $[0,k-2]$, and that $\mathcal{L}_d^N=0$ in the last equality (which implies that we have already picked the Lagrangian \eqref{DiscLag:Dyn}). Equivalent arguments lead to $\sum_{i=k+1}^{N-1}\der\mathcal{L}_d^i/\der y_k=0$ in order to prove \eqref{Sums:b}.

Recalling that $k=1,...,N-1$, from \eqref{Sums:a} we get:
\[
\begin{split}
\frac{\der}{\der x_k}\sum_{i=k-1}^{N-1}\mathcal{L}_d^i=\frac{\der}{\der x_k}\sum_{i=0}^{N}\mathcal{L}_d^i=^{_{1}}&\sum_{i=1}^{N-1}\Big(D_1L_d(x_i,x_{i+1})+D_2L_d(x_{i-1},x_{i})\Big)\delta_{ik}-h\frac{\der}{\der x_k}\sum_{i=1}^{N}\Delta_{-}^{\alpha}x_i\rho\Delta_{+}^{\alpha}y_i\\
=&\sum_{i=1}^{N-1}\Big(D_1L_d(x_i,x_{i+1})+D_2L_d(x_{i-1},x_{i})\Big)\delta_{ik}-h\frac{\der}{\der x_k}\sum_{i=1}^{N-1}x_i\rho\Delta_{+}^{\alpha}\Delta_{+}^{\alpha}y_i\\
=^{_{2}}&\sum_{i=1}^{N-1}\Big(D_1L_d(x_i,x_{i+1})+D_2L_d(x_{i-1},x_{i})-h\,\rho\,\Delta_{+}^{\alpha}\Delta_{+}^{\alpha}y_i\Big)\delta_{ik}\\
=&\quad \,\,\,\,\,\,\,\,D_1L_d(x_k,x_{k+1})+D_2L_d(x_{k-1},x_{k})-h\,\rho\,\Delta_{+}^{\alpha}\Delta_{+}^{\alpha}y_k,
\end{split}
\]
where $\delta_{ij}$ represents the Kronecker delta. In the previous computation, right hand side of $=^{_{1}}$, we have taken into account that
\[
\frac{\der}{\der x_k}\sum_{i=0}^{N}\Delta_{-}^{\alpha}x_i\,\rho\,\Delta_{+}^{\alpha}y_i=\frac{\der}{\der x_k}\sum_{i=1}^{N}\Delta_{-}^{\alpha}x_i\,\rho\,\Delta_{+}^{\alpha}y_i
\]
given that $\der(\Delta_{-}^{\alpha}x_0\,\rho\,\Delta_{+}^{\alpha}y_0)/\der x_k=0$ for $k=1,...,N-1$. Moreover, in the right hand side of $=^{_{2}}$ we have used \eqref{DIParts:b}, $\der (x_0\,\rho\,\Delta_{+}^{\alpha}\Delta_{+}^{\alpha}y_0)/\der x_k=0$ and $\der x_i/\der x_k=\delta_{ij}$ for $k$ in the range $[1,N-1]$.

Using equivalent arguments, it can be proven from \eqref{Sums:b} that:
\[
\begin{split}
\frac{\der}{\der y_k}\sum_{i=0}^{k}\mathcal{L}_d^i=\frac{\der}{\der y_k}\sum_{i=0}^{N-1}\mathcal{L}_d^i=&\sum_{i=1}^{N-1}\Big(D_1L_d(y_i,y_{i+1})+D_2L_d(y_{i-1},y_{i})-h\,\rho\,\Delta_{-}^{\alpha}\Delta_{-}^{\alpha}x_i\Big)\delta_{ik}\\
=&\hspace{0.5cm}\,\,\,\,\,\,\,\, D_1L_d(y_k,y_{k+1})+D_2L_d(y_{k-1},y_{k})-h\,\rho\,\Delta_{-}^{\alpha}\Delta_{-}^{\alpha}x_k,
\end{split}
\]
where again \eqref{DIParts:b} and the range $k=1,...,N-1$ are used.

Adding both terms together and equating the sum to 0, which accounts for \eqref{DELEqs:NS}, we obtain:
\[
\begin{split}
&D_1L_d(x_k,x_{k+1})+D_2L_d(x_{k-1},x_{k})-h\,\rho\,\Delta_{+}^{\alpha}\Delta_{+}^{\alpha}y_k\\
&\,\,+D_1L_d(y_k,y_{k+1})+D_2L_d(y_{k-1},y_{k})-h\,\rho\,\Delta_{-}^{\alpha}\Delta_{-}^{\alpha}x_k=0.
\end{split}
\]
Swapping the fractional terms and equating each block to 0, we directly obtain \eqref{RestDEL:Dyn} as a sufficient condition for \eqref{DELEqs:NS} to hold.
\end{proof}
We note that equations \eqref{RestDEL:Dyn} are the natural discrete version of \eqref{ResFracELDamped}.

\begin{remark}\label{Alphak+1}
Observe that a term $h\Delta_{-}^{\alpha}x_{k+1}\Delta_{+}^{\alpha}y_{k+1}$ is admissible in \eqref{DiscLag:Dyn} according to the definition \eqref{DiscLag}, leading to a term $h\sum_{k=0}^{N-1}\Delta_{-}^{\alpha}x_{k+1}\Delta_{+}^{\alpha}y_{k+1}$ in the action sum \eqref{DiscAct}. However, it provides the same discrete dynamics as  $\Delta_{-}^{\alpha}x_{k}\Delta_{+}^{\alpha}y_{k}$, according to \eqref{DIParts:c}, which makes it redundant. Further terms, as those described in Remark \ref{MixedIndices}, are meaningless since the asymetric integration by parts is not defined for them. \hfill $\diamond$
\end{remark}

In the next result we prove that the $x$ and $y$ dynamics in \eqref{RestDEL:Dyn} are also related under inversion of time at a discrete level.

\begin{proposition}\label{InvDiscTime}
Given $y_k:=x_{N-k}$, then \eqref{RestDEL:Dyn:b} is \eqref{RestDEL:Dyn:a} in reversed discrete time.
\end{proposition}
\begin{proof}
We define the reversed discrete time as $\tilde k:=N-k$, such that $\tilde k=N,...,0$, for $k=0,...,N$. Given that, we observe that, under inversion of time:
\[
\begin{split}
D_1L_d(y_k,y_{k+1})\,+\,D_2L_d(y_{k-1},y_{k})&=D_1L_d(x_{N-k},x_{(N-k)+1})\,+\,D_2L_d(x_{(N-k)-1},x_{N-k})\\
&=D_1L_d(x_{\tilde k},x_{\tilde k+1})\,+\,D_2L_d(x_{\tilde k-1},x_{\tilde k}).
\end{split}
\]
On the other hand:
\[
\begin{split}
\Delta_{+}^{\alpha}\Delta_{+}^{\alpha}y_k=\frac{1}{h^{2\alpha}}\sum_{n=0}^{N-k}\alpha_n\sum_{p=0}^{N-k-n}\alpha_p\,y_{k+n+p}&=\frac{1}{h^{2\alpha}}\sum_{n=0}^{(N-k)}\alpha_n\sum_{p=0}^{(N-k)-n}\alpha_p\,x_{(N-k)-n-p}\\
&=\frac{1}{h^{2\alpha}}\sum_{n=0}^{\tilde k}\alpha_n\sum_{p=0}^{\tilde k-n}\alpha_p\,x_{\tilde k-n-p}=\Delta_{-}^{\alpha}\Delta_{-}^{\alpha}x_{\tilde k}.
\end{split}
\] 
Multiplying $\Delta_{-}^{\alpha}\Delta_{-}^{\alpha}x_{\tilde k}$ by $h\,\rho$, adding it to $D_1L_d(x_{\tilde k},x_{\tilde k+1})\,+\,D_2L_d(x_{\tilde k-1},x_{\tilde k})$ and equating the sum to 0, the claim holds.
\end{proof}

As advanced in Remark \ref{AllVariables}, in spite all variables are present at the same time in \eqref{DELEqs:NS}, which {\it a priori} makes necessary to solve them all simultaneously, we can find particular and meaningful Lagrangian functions \eqref{DiscLag:Dyn} such that the discrete dynamics \eqref{RestDEL:Dyn:a} (we restrict ourselves to the $x$-system since we have just proved that $y$ can be interpreted as $x$ in reversed time) is provided by a discrete flow.  That is established in the following algorithm, accounting for the definition of the FVIs:

\begin{algorithm}{\rm Fractional Variational Integrator Scheme}
\label{FractionalAlgorithm}
\begin{algorithmic}[1]
\State {\bf Initial data}: $N,\, h,\,\alpha,\,\rho,\, x_0,\, p_{x_0}.$
 \State {\bf solve for} $x_1$ {\bf from} $p_{x_0}=I(x_0,x_1).$
\State {\bf Initial points:} $x_0,\,x_1.$
    \For {$k= 1: N-1$} 
    
\hspace{-0.6cm} {\bf solve for} $x_{k+1}$ {\bf from} $D_1L_d(x_k,x_{k+1})+D_2L_d(x_{k-1},x_{k})-h\,\rho\,\Delta_{-}^{\alpha}\Delta_{-}^{\alpha}x_k=0$
    \EndFor
    \State  {\bf Output:} $(x_2,...,x_N).$
\end{algorithmic}
  \end{algorithm}
Observe that the initialisation condition in Step 2, i.e. $p_{x_0}=I(x_0,x_1)$, has to be properly determined. This will be object of discussion when defining the fractional discrete Legendre transformation in \S\ref{DLTra}.

\subsection{Discrete mechanical Lagrangian}
Now, let us pick the discrete Lagrangian
\begin{equation}\label{DiscMechLag}
L_d(z_k,z_{k+1}):=\frac{1}{2h}(z_{k+1}-z_k)\,m\,(z_{k+1}-z_k)-h\,U(S^{\kappa}z_k), 
\end{equation}
where $S^{\kappa}z_k$ is given in \eqref{KappaRule}, as the usual discretisation of the mechanical Lagrangian \eqref{MechLag}. In this case, \eqref{RestDEL:Dyn} read
\begin{equation}\label{MechDynFracDamp:Disc}
\begin{split}
m\frac{x_{k+1}-2x_k+x_{k-1}}{h^2}&+\kappa\nabla U(S^{\kappa}x_k)+(1-\kappa)\nabla U(S^{\kappa}x_{k-1})+\rho\,\Delta_{-}^{\alpha}\Delta_{-}^{\alpha}x_k=0,\\
m\frac{y_{k+1}-2y_k+y_{k-1}}{h^2}&+\kappa\nabla U(S^{\kappa}y_k)+(1-\kappa)\nabla U(S^{\kappa}y_{k-1})+\rho\,\Delta_{+}^{\alpha}\Delta_{+}^{\alpha}y_k=0,
\end{split}
\end{equation}
where we have divided both sides by $h$. We observe that these equations are a discretisation in finite differences of  \eqref{MechDynFracDamp}. 

According to what happens in the continous case, i.e. $D^{2\alpha}_{-}= d/dt$ when $\alpha= 1/2$, we expect a particular discretisation of the total time derivative from the term $\Delta_{-}^{\alpha}\Delta_{-}^{\alpha}x_k$. This is proven in the following result.

\begin{lemma}\label{LemmaDisc}
For $\alpha=1/2$ and $k=1,...,N-1$,
\[
h\,\Delta_{-}^{_{1/2}}\Delta_{-}^{_{1/2}}x_k=\sum_{n=0}^k\alpha_n\big|_{_{\alpha=1/2}}\sum_{p=0}^{k-n}\alpha_p\big|_{_{\alpha=1/2}}x_{k-n-p}=x_k-x_{k-1}.
\]
\end{lemma}
\begin{proof} According to \eqref{AlphaDef} we have that $\alpha_0=1$ and $\alpha_1=-1/2$ for $\alpha=1/2$, leading, after expanding the summations, to
\[
\begin{split}
\sum_{n=0}^k\alpha_n\sum_{p=0}^{k-n}\alpha_px_{k-n-p}=(x_k-x_{k-1})+\sum_{n=2}^{k}2\alpha_nx_{k-n}+\sum_{n=1}^k\alpha_n\sum_{p=1}^{k-n}\alpha_px_{k-n-p}.
\end{split}
\]
In this expansion, the value of the sum when $k=0$, explicited in Remark \ref{SumEndpoint}, has been taken into account. The claim automatically holds for $k=1$. For $k\geq 2$, arranging the sum indices we see that the previous expression can be rewritten as
\begin{equation}\label{SumExp}
\begin{split}
\sum_{n=0}^k\alpha_n\sum_{p=0}^{k-n}\alpha_px_{k-n-p}=&(x_k-x_{k-1})\\
+&\sum_{s=2}^{r}\beta_0^sx_{k-s}+\sum_{l=0}^{k-(r+1)}\beta_l^{r+1}x_{k-(r+1)-n}+\sum_{n=r}^k\alpha_n\sum_{p=1}^{k-n}\alpha_px_{k-n-p},
\end{split}
\end{equation}
where, for a fixed $k$, we set $r=k-1$ and 
\begin{equation}\label{Betas}
\beta_l^j=2\alpha_{l+j}+\sum_{i=1}^{j-1}\alpha_i\alpha_{l+j-i},
\end{equation}
(it is apparent that $j$ is not a power but a superindex). For a fixed $k=\tilde k$, \eqref{SumExp} acquires the form
\[
\begin{split}
&\sum_{n=0}^{\tilde k}\alpha_n\sum_{p=0}^{\tilde k-n}\alpha_px_{\tilde k-n-p}=(x_{\tilde k}-x_{\tilde k-1}) +\beta_0^2x_{\tilde k-2}+\beta_0^3x_{\tilde k-3}+\cdot\cdot\cdot+\beta_0^{\tilde k-2}x_2+\beta_0^{\tilde k-1}x_1+\beta_0^{\tilde k}x_0.
\end{split}
\]
According to this, it is enough to prove that $\beta_0^j=0$ for any $j$, for which we proceed by induction. From \eqref{Betas} and \eqref{AlphaDef}, it follows that $\beta^2_0=2\alpha_2+\alpha_1\alpha_1$, which vanishes for $\alpha=1/2$. Taking this as the first induction step, it is enough to prove that $\beta^{j+1}_0=0$ assuming that $\beta^j_0=0$. This is shown next:
\[
\begin{split}
&\beta^{j+1}_0=2\alpha_{j+1}+\sum_{i=1}^{j}\alpha_i\alpha_{j+1-i}=2\alpha_{j+1}+\sum_{i=1}^{r-1}\alpha_i\alpha_{r-i}\\
&\hspace{3cm} \quad\quad=2\alpha_{j+1}-2\alpha_r+2\alpha_r+\sum_{i=1}^{r-1}\alpha_i\alpha_{r-i}=2\alpha_{j+1}-2\alpha_r=0,
\end{split}
\]
where we have set $r=j+1.$ Hence the claim follows.
\end{proof}

Using similar arguments, one can prove that
\[
h\,\Delta_{+}^{_{1/2}}\Delta_{+}^{_{1/2}}y_k=\sum_{n=0}^{N-k}\alpha_n\big|_{_{\alpha=1/2}}\sum_{p=0}^{N-k-n}\alpha_p\big|_{_{\alpha=1/2}}y_{k+n+p}=-(y_{k+1}-y_k).
\]
It follows straightforwardly that
\begin{equation}\label{Ffrac:OneHalf}
\Delta_{-}^{_{1/2}}\Delta_{-}^{_{1/2}}x_k=\frac{x_k-x_{k-1}}{h},\,\quad \Delta_{+}^{_{1/2}}\Delta_{+}^{_{1/2}}y_k=-\frac{y_{k+1}-y_{k}}{h},
\end{equation}
showing that $\Delta_{-}^{_{1/2}}\Delta_{-}^{_{1/2}}x_k$, $\Delta_{+}^{_{1/2}}\Delta_{+}^{_{1/2}}y_k$ are the backward and forward (up to a minus sign) difference operators, respectively; thus order one approximations of the velocity. This leads to the following corollary of Theorem \ref{Theo:DiscEqs}.

\begin{corollary}\label{Corolario}
{\it Given $\alpha=1/2$, then \eqref{RestDEL:Dyn:a} is equivalent to the forced discrete Euler-Lagrange equations \eqref{ForcedDEL} when $f^{-}_{L_d}(x_k,x_{k+1})=0$, $f^+_{L_d}(x_k,x_{k+1})=-\rho\,(x_{k+1}-x_k).$ }
\end{corollary}

\begin{proof}
The result follows straightforwardly by replacing \eqref{Ffrac:OneHalf} (first relationship) in \eqref{RestDEL:Dyn:a} and comparing with \eqref{ForcedDEL} for $f^{-}_{L_d}(x_k,x_{k+1})=0$, $f^+_{L_d}(x_k,x_{k+1})=-\rho\,(x_{k+1}-x_k).$
\end{proof}
An equivalent result can be obtained for the $y$-mirror system setting $f_{L_d}^-(y_k,y_{k+1})=\rho(y_{k+1}-y_k)$ and $f_{L_d}^{+}(y_k,y_{k+1})=0$. This discussion makes explicit the relationship between the discretisation of the Lagrange-d'Alembert principle (\S\ref{DiscLdAlem}, \eqref{ForcedDEL}) and the restricted Hamilton's principle developed in this work, that we display in Figure \ref{Diagram1} (where we omit the $y$-system for sake of simplicity).


\begin{figure}[!htb]
\includegraphics[scale=0.42]{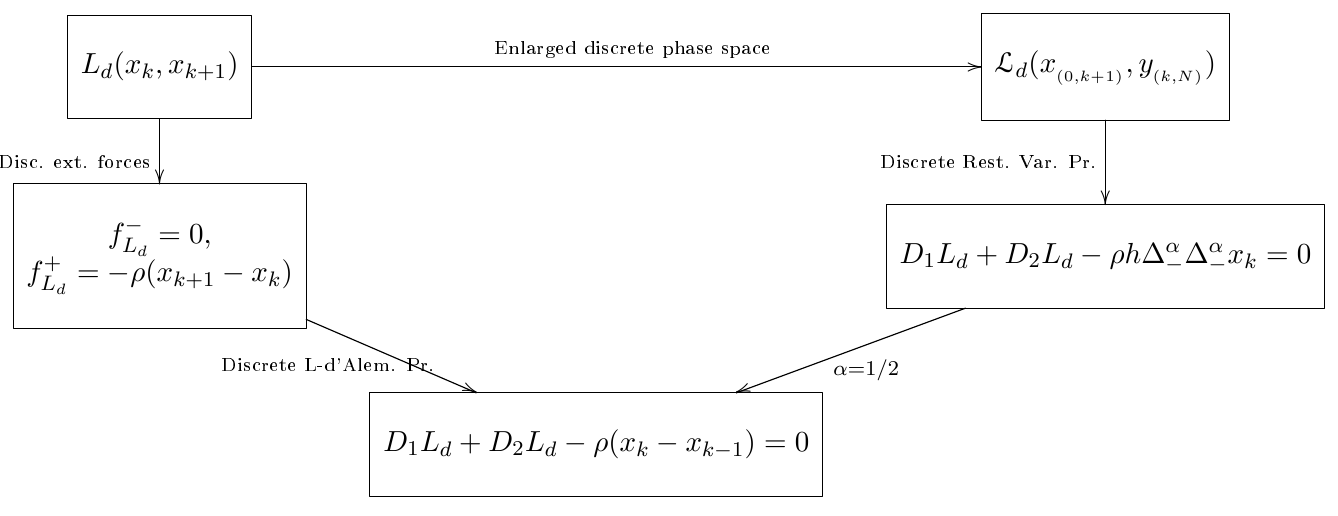}
\caption{In the diagram it is reflected that the Lagrange-d'Alembert principle for a general discrete Lagrangian and a particular set of discrete forces, produce the same discrete forced equations as enlarging the discrete phase space, then applying the discrete restricted Hamilton's principle for $\mathcal{L}_d^k$ \eqref{DiscLag:Dyn}, and then setting $\alpha=1/2$.}
\label{Diagram1}
\end{figure}

\subsection{Discrete Legendre transformation}\label{DLTra}

The main guidelines to construct the discrete Legendre transformation in the fractional case are the following:
\begin{itemize}
\item[1.] As in the usual discrete mechanics \S\ref{DiscreteFramework}, we want to reproduce the discrete Lagrangian dynamics \eqref{RestDEL:Dyn} through {\it momentum matching}.
\item[2.] We seek for a fair discretisation of the fractional Hamilton equations in the case of mechanical Hamiltonians \eqref{FracPartiHam}. 
\item[3.] We intend to obtain initialisation condition (Step 2) for Algorithm \ref{FractionalAlgorithm}.
\end{itemize}
According to this, we provide the following definition of the discrete Legendre transformation. Previously, we introduce the intermediate variables:
\[
x_k^{\alpha}:=\Delta^{\alpha}_{-}x_k,\quad\quad y_k^{\alpha}:=\Delta^{\alpha}_{+}y_k,
\]
and replace them into the Lagrangian \eqref{DiscLag:Dyn} in order to establish the following definition
\begin{equation}\label{DiscLagInterm}
\begin{split}
\mathcal{L}_d^k(x_{(0,k+1)},y_{(k,N)})&:=\tilde{\mathcal{L}}_d^k(x_k,x_{k+1},y_k,y_{k+1},x_k^{\alpha},y_k^{\alpha},x_{k+1}^{\alpha},y_{k+1}^{\alpha}).
\end{split}
\end{equation}
For the sake of generality, we allow the presence of $x_{k+1}^{\alpha}=\Delta_{-}^{\alpha}x_{k+1}$ and $y_{k+1}^{\alpha}=\Delta_{+}^{\alpha}y_{k+1}$, which are admissible as discussed in Remark \ref{Alphak+1}.

\begin{definition}\label{Def:DLT}
{\it Given the discrete Lagrangian \eqref{DiscLagInterm}, we define the} discrete Legendre transformation {\it by:}
\begin{subequations}\label{LegTrans}
\begin{align}
\begin{bmatrix}
p_{x_k}^{-}\\
p_{y_k}^{-}
\end{bmatrix}&=
-\begin{bmatrix}
\,\,D_{x_k}\tilde{\mathcal{L}}_d^k\,\,\\
\,\,D_{y_k}\tilde{\mathcal{L}}_d^k\,\,
\end{bmatrix}
\,\,\,\,-\,\begin{bmatrix}
\Delta_{-}^{\alpha} & \Delta_{+}^{\alpha}
\end{bmatrix}
\begin{bmatrix}
0&1\\
1&0
\end{bmatrix}
\begin{bmatrix}
\,\,D_{x_k^{\alpha}}\tilde{\mathcal{L}}_d^k\,\,\\
\,\,D_{y_k^{\alpha}}\tilde{\mathcal{L}}_d^k\,\,
\end{bmatrix}, \label{LegTrans:a}\\
\begin{bmatrix}
p_{x_{k+1}}^{+}\\
p_{y_{k+1}}^{+}
\end{bmatrix}&=\,\,\,\,\,
\begin{bmatrix}
\,\,D_{x_{k+1}}\tilde{\mathcal{L}}_d^k\,\,\\
\,\,D_{y_{k+1}}\tilde{\mathcal{L}}_d^k\,\,
\end{bmatrix}+\begin{bmatrix}
\Delta_{-}^{\alpha} & \Delta_{+}^{\alpha}
\end{bmatrix}
\begin{bmatrix}
0&1\\
1&0
\end{bmatrix}
\begin{bmatrix}
\,\,D_{x_{k+1}^{\alpha}}\tilde{\mathcal{L}}_d^k\,\,\\
\,\,D_{y_{k+1}^{\alpha}}\tilde{\mathcal{L}}_d^k\,\,
\end{bmatrix}
,\label{LegTrans:b}\\
\begin{bmatrix}
p_{x_{k}}^{\alpha}\\
p_{y_{k}}^{\alpha}
\end{bmatrix}&=\,\,\,\,\,
\begin{bmatrix}
\,\,D_{x_{k}^{\alpha}}\tilde{\mathcal{L}}_d^k\,\,\\
\,\,D_{y_{k}^{\alpha}}\tilde{\mathcal{L}}_d^k\,\,
\end{bmatrix} \label{LegTrans:c}
\end{align}
\end{subequations}
{\it where we consider the row matrix $[\Delta_{-}^{\alpha}\,\,  \Delta_{+}^{\alpha}]$ in the sense of operators.}
\end{definition}
\begin{proposition}\label{MMPropo}
Given the discrete Legendre transformation in Definition \eqref{Def:DLT}, the particular discrete Lagrangian $\tilde{\mathcal{L}}_d^k=L_d(x_k,x_{k+1})+L_d(y_k,y_{k+1})-h\,x_{k+1}^{\alpha}\,\,\rho\,\,y_{k+1}^{\alpha}$ and the intermediate variables $x_k^{\alpha}:=\Delta^{\alpha}_{-}x_k,\,\, y_k^{\alpha}:=\Delta^{\alpha}_{+}y_k,$ the following statements are true:
\begin{itemize}
\item[1.] The momentum matching condition $p_{x_k}^{-}=p_{x_k}^{+}$, $p_{y_k}^{-}=p_{y_k}^{+}$, is equivalent to the discrete Lagrangian dynamics \eqref{RestDEL:Dyn}.
\item[2.] Given the mechanical Hamiltonian \eqref{MechHamilto} and $L_d(z_k,z_{k+1})=\frac{1}{2h}(z_{k+1}-z_k)\,m\,(z_{k+1}-z_k)-h\,U(z_{k+1})$, i.e. we pick $\kappa=0$ in \eqref{KappaRule},\eqref{DiscMechLag} , then \eqref{LegTrans} provides a discretisation of \eqref{FracPartiHam}.
\item[3.] Under the hypotheses of Statement 2, when $\alpha= 1/2$ \eqref{LegTrans} provides a discretisation of \eqref{HamLinearDamp}, i.e. the Hamiltonian dynamics with linear damping.
\end{itemize}
\end{proposition}

\begin{proof}
Statement 1: from \eqref{LegTrans:a}, \eqref{LegTrans:b} and $\tilde{\mathcal{L}}_d^k=L_d(x_k,x_{k+1})+L_d(y_k,y_{k+1})-h\,x_{k+1}^{\alpha}\rho y_{k+1}^{\alpha}$ it follows:
\[
\begin{split}
\begin{bmatrix}
p_{x_k}^{-}\\
p_{y_k}^{-}
\end{bmatrix}&=
-\begin{bmatrix}
D_{1}L_d(x_k,x_{k+1})\\
D_{1}L_d(y_k,y_{k+1})
\end{bmatrix},\\
\begin{bmatrix}
p_{x_{k+1}}^{+}\\
p_{y_{k+1}}^{+}
\end{bmatrix}&=\,\,\,\,\,
\begin{bmatrix}
D_{2}L_d(x_k,x_{k+1})\\
D_{2}L_d(y_k,y_{k+1})
\end{bmatrix}-h\,\rho\,\begin{bmatrix}
\Delta^{\alpha}_{-}\,x_{k+1}^{\alpha}\\
\Delta^{\alpha}_{+}\,y_{k+1}^{\alpha}
\end{bmatrix}=\begin{bmatrix}
D_{2}L_d(x_k,x_{k+1})\\
D_{2}L_d(y_k,y_{k+1})
\end{bmatrix}
-h\,\rho\,\begin{bmatrix}
\Delta^{\alpha}_{-}\Delta^{\alpha}_{-}\,x_{k+1}\\
\Delta^{\alpha}_{+}\Delta^{\alpha}_{+}\,y_{k+1}
\end{bmatrix}.
\end{split}
\]
Now setting the momentum matching condition $p_{x_k}^{-}=p_{x_k}^{+}$, $p_{y_k}^{-}=p_{y_k}^{+}$ it is straightforward to obtain \eqref{RestDEL:Dyn}.

Statement 2: under the hypotheses, we have that \eqref{FracPartiHam} read

\begin{equation}\label{HamFinal}
\begin{split}
\dot x&=m^{-1}p_x,\quad\quad D^{\alpha}_{-}x=-\rho^{-1}\,p_y^{\alpha}, \quad\quad \dot p_x=-\nabla U(x)+D^{\alpha}_{-}\,p^{\alpha}_y,\\
\dot y&=m^{-1}p_y, \quad\quad
 D^{\alpha}_{+}y=-\rho^{-1}\,p_x^{\alpha}, \quad\quad \dot p_y=-\nabla U(y)+D^{\alpha}_{+}\,p^{\alpha}_x.
\end{split}
\end{equation}
On the other hand, from \eqref{LegTrans} we get

\[
\begin{split}
\begin{bmatrix}
p_{x_k}\\
p_{y_k}
\end{bmatrix}&=
\,\,\,\,\,\begin{bmatrix}
\,\,m\frac{x_{k+1}-x_k}{h}\,\,\\
\,\,m\frac{y_{k+1}-y_k}{h}\,\,
\end{bmatrix}, \\
\begin{bmatrix}
p_{x_{k+1}}\\
p_{y_{k+1}}
\end{bmatrix}&=\,\,\,\,\,\,
\begin{bmatrix}
\,\,m\frac{x_{k+1}-x_k}{h}-h\nabla U(x_{k+1})\,\,\\
\,\,m\frac{y_{k+1}-y_k}{h}-h\nabla U(y_{k+1})\,\,
\end{bmatrix}-\begin{bmatrix}
h\,\rho\,\Delta^{\alpha}_{-}\,x_{k+1}^{\alpha}\\
h\,\rho\,\Delta^{\alpha}_{+}\,y_{k+1}^{\alpha}
\end{bmatrix},\\
\begin{bmatrix}
p_{x_{k+1}}^{\alpha}\\
p_{y_{k+1}}^{\alpha}
\end{bmatrix}&=
-\begin{bmatrix}
\,\,h\,\rho\,y^{\alpha}_{k+1}\,\,\\
\,\,h\,\rho\,x^{\alpha}_{k+1}\,\,
\end{bmatrix} = -\begin{bmatrix}
\,\,h\,\rho\,\Delta^{\alpha}_{+}\,y_{k+1}\,\,\\
\,\,h\,\rho\,\Delta^{\alpha}_{-}\,x_{k+1}\,\,
\end{bmatrix}.
\end{split}
\]
(Observe that, in the last equation, we have taken a $k$-step forward in \eqref{LegTrans:c}). From this, splitting $x$ and $y$ sides and rearranging terms, we obtain
\begin{equation}\label{DiscHamFinal}
\begin{split}
x_{k+1}&=x_k+h\,m^{-1}p_{x_{k}},\,\,\,\,\Delta^{\alpha}_{-}\,x_{k+1}=-h^{-1}\rho^{-1}p_{y_{k+1}}^{\alpha},\,\,\,\,p_{x_{k+1}}=p_{x_k}-h\nabla U(x_{k+1})+\Delta^{\alpha}_{-}p^{\alpha}_{y_{k+1}},\\
y_{k+1}&=y_k+h\,m^{-1}p_{y_{k}},\,\,\,\, \Delta^{\alpha}_{+}\,y_{k+1}=-h^{-1}\rho^{-1}p_{x_{k+1}}^{\alpha},\,\,\,\,\,p_{y_{k+1}}=p_{y_k}-h\nabla U(y_{k+1})+\Delta^{\alpha}_{+}p^{\alpha}_{x_{k+1}},
\end{split}
\end{equation}
which are a natural discretisation of \eqref{HamFinal}.

Statement 3: we focus on the $x$-system in \eqref{DiscHamFinal}. Replacing the fractional equation into the third one, we obtain
\[
\begin{split}
x_{k+1}=x_k+h\,m^{-1}p_{x_{k}},\quad\quad p_{x_{k+1}}&=\,\,\,p_{x_k}-h\nabla U(x_{k+1})-h\,\rho\,\Delta^{_{1/2}}_{-}\Delta^{_{1/2}}_{-}x_{k+1}\\
&=^{_1}p_{x_k}-h\nabla U(x_{k+1})-\rho\,(x_{k+1}-x_{k})\\
&=^{_2}p_{x_k}-h\nabla U(x_{k+1})-h\,\rho\,m^{-1}\,p_{x_k},
\end{split}
\]
where, in $=^{_1}$ we have employed Lemma \ref{LemmaDisc} and in $=^{_2}$ we have employed the discrete $x$-dynamics, i.e. $x_{k+1}=x_{k}+h\,m^{-1}p_{x_{k}}$. Naturally, the previous equations are a discretisation of \eqref{HamLinearDamp}.
\end{proof}

Furthermore, the discrete Legendre transformation in Definition \ref{Def:DLT} provides a initialisation step $p_{x_0}=I(x_0,x_1)$ (Step 2) for Algorithm \ref{FractionalAlgorithm}. Namely, the $x$-part of \eqref{LegTrans:a} reads
\begin{equation}\label{InitialCond}
p_{x_k}=-D_{x_k}\tilde{\mathcal{L}}_d^k-\Delta_{-}^{\alpha} D_{x_k}^{\alpha}\tilde{\mathcal{L}}_d^k,
\end{equation}
which, for $k=0$, only involves $p_{x_0},x_0$ and $x_1$. For the particular $\tilde{\mathcal{L}}_d^k$ in the theorem above, the initial condition reads  $p_{x_0}=-D_{1}L_d(x_0,x_{1})$.

\begin{remark}
The matrix $\begin{bmatrix}0&1\\1&0\end{bmatrix}$ in second term of the right hand side of \eqref{LegTrans:a} and \eqref{LegTrans:b} obeys to the necessity of decoupling $x$ and $y$ dynamics at the discrete level, which we achieved by restricting the variations and setting the critical conditions as (only) sufficient in the Lagrangian side, as shown in Proposition \ref{FullLagProposition} and Remark \ref{RemarkLagrangianMixing}. In other words, it can be considered as a discrete Hamiltonian consequence of the restricted Hamilton's principle. \hfill $\diamond$
\end{remark}

\begin{remark}
It is interesting to note that the result remains the same for any $\tilde{\mathcal{L}}_d^k=L_d(x_k,x_{k+1})+L_d(y_k,y_{k+1})-\tilde\kappa\, h\,x_{k}^{\alpha}\,\,\rho\,\,y_{k}^{\alpha}-(1-\tilde\kappa)\, h\,x_{k+1}^{\alpha}\,\,\rho\,\,y_{k+1}^{\alpha}$,
with $\tilde\kappa\in[0,1]$, which is a way of rephrasing Remark \ref{Alphak+1}. However, the presence of $x^{\alpha}_k$ turns the initial condition \eqref{InitialCond} meaningless from a physical point of view, which makes convenient setting $\tilde\kappa=0$. In that case, the pick of $L_d(x_k,x_{k+1})$, which implies a particular choice of  $\kappa$ in \eqref{KappaRule} and 	\eqref{DiscMechLag}, leads to different discretisations of \eqref{HamFinal} and \eqref{HamLinearDamp}. We remark that the chosen one ($\kappa=0$) preserves the semi-implicitness of variables $x,p_x$ of the symplectic-Euler methods \cite{SS} for \eqref{HamLinearDamp}; say: the final integrator is explicit in the variable $p_x$ and implicit in the variable $x$.
\hfill $\diamond$
\end{remark}

\section{Numerical simulations}\label{Simu}

As a first test example, we employ the linearly damped harmonic oscillator with  potential function $U(x)=\mathfrak{c}x^2/2$ and dynamical equation (exactly solvable):
\begin{equation}\label{LinearDamping}
m\ddot x+\mathfrak{c}x+\rho\dot x=0,
\end{equation}
with $m=1$, $\mathfrak{c}=1$, $\rho=0.2$, $x(0)=1$ and $p_x(0)=\dot x(0)=0.5$ in the simulations.

\begin{figure}[!htb]
\includegraphics[scale=0.7]{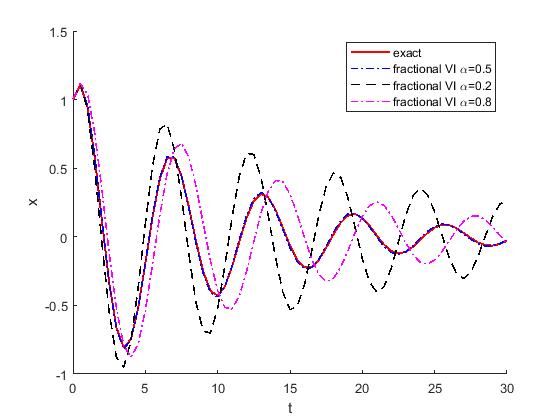}
\caption{Fractional Variational Integrators (FVIs), determined by  Algorithm \ref{FractionalAlgorithm}, for several values of $\alpha$, with $h=0.5$ and $N=30$. The Exact (red) line corresponds to the exact solution of \eqref{LinearDamping} for the given set of parameters and initial conditions.}
\label{PlotAlpha}
\end{figure}
\vspace{2cm}

In Figure \ref{PlotAlpha}, we show the outcome of Algorithm \ref{FractionalAlgorithm} with initial condition $p_{x_0}=-D_{1}L_d(x_0,x_{1})$ according to \eqref{InitialCond}, for several $\alpha$'s, where we choose $\kappa=1/2$ in \eqref{DiscMechLag} since it provides the midpoint rule for the potential and it is where the maximum local truncation order (namely 2) is achieved in usual low-order variational integrators \cite{MaWe}. We observe that the FVI approximates properly the solution of \eqref{LinearDamping} when $\alpha = 1/2$, which is natural since that is the case when \eqref{MechDynFracDamp:a} $\Flder$ \eqref{LinearDamping} (in other words,  $D^{_{1/2}}_{-}D^{_{1/2}}_{-}= d/dt$). Moreover, according to Corollary \ref{Corolario}, that is also the case when the FVI is equivalent to the Forced Variational Integrator (coming out of the discrete Lagrange-d'Alembert principle), Algorithm \ref{ForcedAlgorithm}, when $f_{L_d}^+(x_k,x_{k+1})=-\rho\,(x_{k+1}-x_k).$  This theoretical agreement is numerically tested (and shown) up to machine rounding error in Figure \ref{PositionFigure}  (Lower-Left plot), for the different implementations of both Algorithms \ref{FractionalAlgorithm} and \ref{ForcedAlgorithm}. We also show the comparison of the FVI to implicit and explicit Euler integrators for \eqref{HamLinearDamp}, choosing a smaller $h$ for the latter, which is necessary to obtain stable simulations for the explicit Euler scheme. In particular, we compare the $x$-trajectories in Figure~\ref{PositionFigure} (Upper plots) and the energy in Figure~\ref{EnergyFigure} (Left and Middle plots), where naturally, we define the continuous and discrete energies by $E(t)=p(t)^2/2m+\mathfrak{c}\,x(t)^2/2$ and $E_k=p_k^2/2m+\mathfrak{c}\,x_k^2/2$ for $k=0,...,N$; respectively. While implicit and explicit Euler artificially gains respectively looses energy, the FVI respects the energy decay due to the dissipation much better and is very close to the exact solution. 

We finally do a numerical convergence study by investigating the global error in both $x$ and $p_x$ variables (Lower-Right plot in Figure~\ref{PositionFigure}) as well as for the energy (Right plot in Figure~\ref{EnergyFigure}). Here, the global error is defined as
\[
\mbox{max}\,|x(t_k)-x_k|,\,\,\,\forall\,\, k,
\]
and equivalently for any other quantity.
For all quantities the convergence is of order $O(h^{0.94})$, i.e.~we obtain a convergence rate of approximately $1$.

\begin{figure}[!htb]
\centering
\begin{tabular}{cc}
\includegraphics[width=7cm]{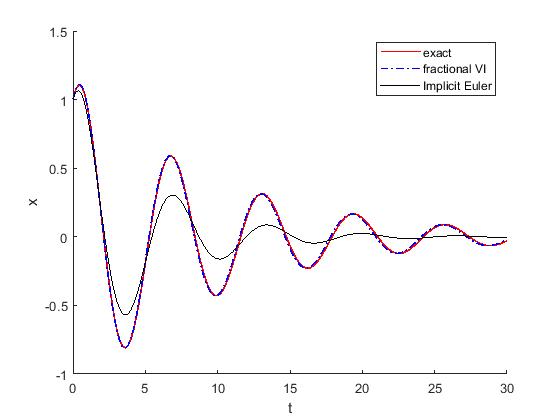}&\includegraphics[width=7cm]{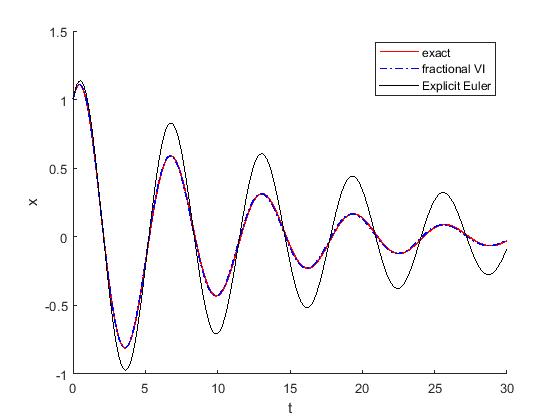}\\
\includegraphics[width=7cm]{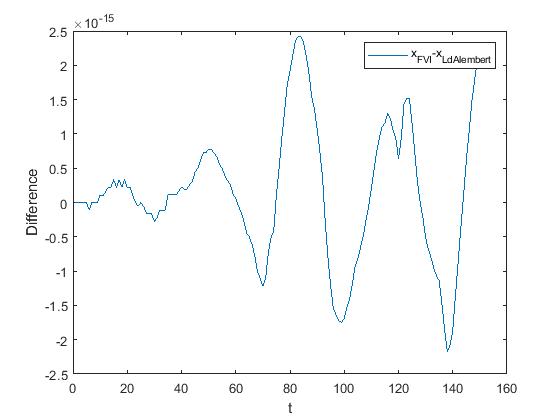} &\includegraphics[width=7cm]{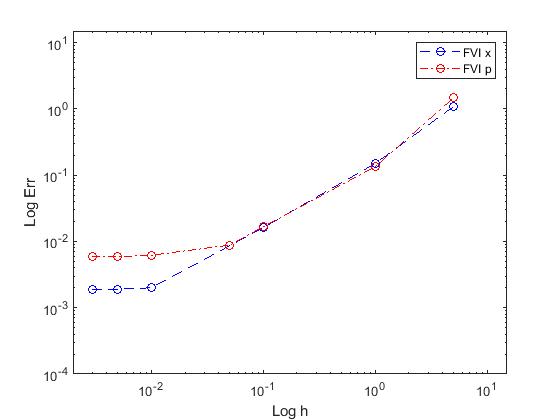}
\end{tabular}
\caption{Upper-Left: FVI vs. implicit Euler for $h=0.2$. Upper-Right: FVI vs. explicit Euler for $h=0.1$. Lower-Left: Difference FVI-Lagrange-d'Alembert integrators: $x_k^{_{FVI}}-x_k^{_{LdA}}$ for $h=0.2$. Lower-Right: Log-Log plot of the global error of $x$ and $p_x$ vs. the time step $h$ for FVI.}
\label{PositionFigure}
\end{figure}
\vspace{10cm}

\begin{figure}[!htb]
\centering
\begin{tabular}{ccc}
\includegraphics[width=5.5cm]{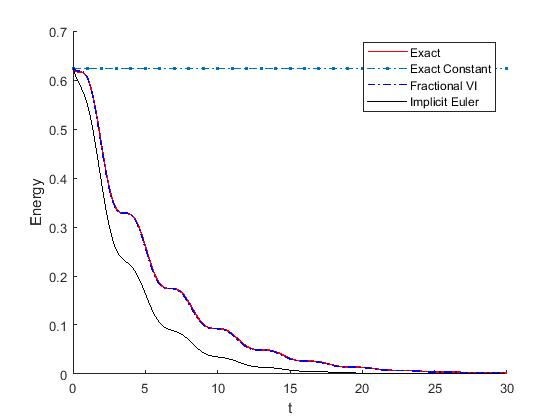}&\includegraphics[width=5.5cm]{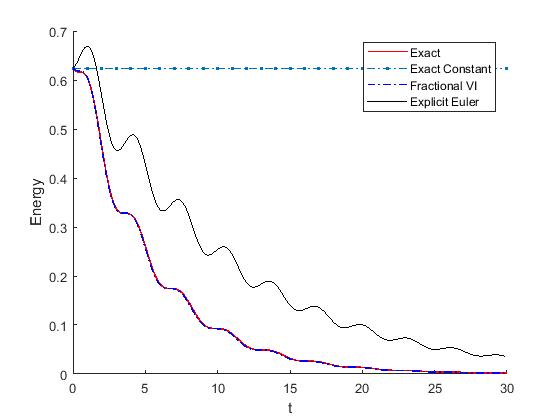}&\includegraphics[width=5.5cm]{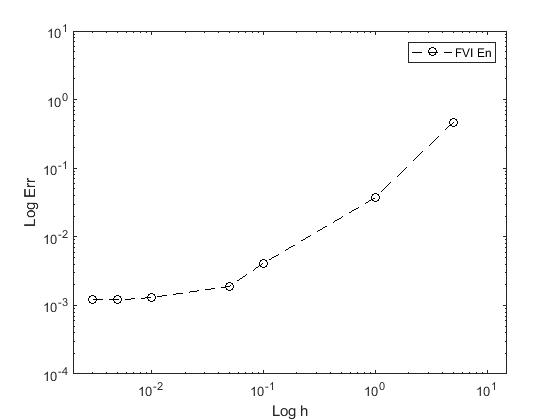}
\end{tabular}
\caption{Left: FVI vs. implicit Euler for $h=0.2$. Middle: FVI vs. explicit Euler for $h=0.1$. Right: Log-Log plot of the global error of the energy vs. the time step $h$ for FVI.}
\label{EnergyFigure}
\end{figure}

As for $\alpha\neq1/2$, we employ the following test example:
\begin{equation}\label{BenchExample}
\ddot x+x+D^{_{3/2}}_{-}x=F(t),\quad F(t)=8\,\,\mbox{for}\,\,0\leq t\leq 1,\,\,F(t)=0\,\,\mbox{for}\,\,t>1,
\end{equation}
and initial conditions $x(0)=\dot x(0)=0$. This equation corresponds to \eqref{MechDynFracDamp:a} when $U(x)=x^2/2$, $m=\rho=1$, $\alpha=3/4$ and we add the inhomogeneous external force $F(t)$ in the right hand side (which can be easily carried out by equating the virtual work of such a force and the variation of the actions \eqref{ContAction} and \eqref{DiscAct}, in continuous and discrete scenarios, respectively, under restricted calculus of variations). Equation \eqref{BenchExample} has an exact solution\footnote{Particularly, it is a so-called {\it inhomogeneous sequential fractional differential equation}, i.e.
\[
\lp D^{n/q}_{-}+a_1D^{(n-1)/q}_{-}+\cdots+a_nD^0_{-}\rp\,x(t)=F(t),
\]
with $(n,q)=(4,2)$, $a_1=a_4=1$, $a_2=a_3=0$ \cite{MiRo}.}, but it is of difficult implementation. For that reason, we employ the benchmark numerical solution designed for MatLab in \cite{Vinagre}. We display the performance of the FVI versus the benchmark solution (with a much smaller time step) (Left plot in Figure~\ref{BenchFigure}) and observe a global convergence of order $O(h^{0.97})$ (Right plot in Figure~\ref{BenchFigure}), i.e.~a convergence rate of approximately $1$ as for the $\alpha=1/2$ case.
\begin{figure}[!htb]
\centering
\begin{tabular}{cc}
\includegraphics[width=7.5cm]{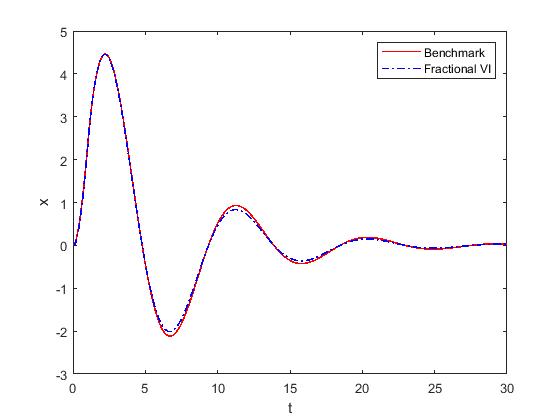}&\includegraphics[width=7.5cm]{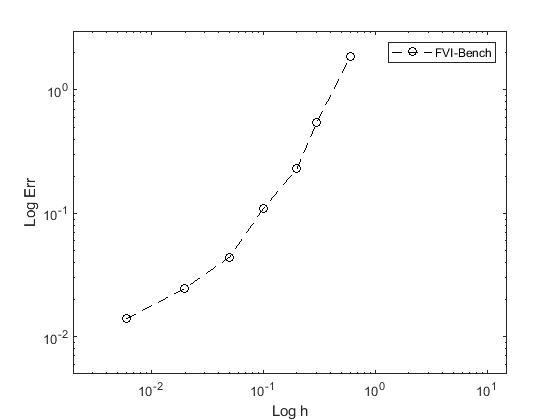}
\end{tabular}
\caption{Left: Benchmark solution ($h=5\times 10^{-3}$) vs. FVI ($h=10^{-1}$). Right: Log-Log plot of the global error of $x$ vs. the time step $h$ (with the benchmark solution as reference).}
\label{BenchFigure}
\end{figure}

\section{Conclusions}

We have developed a restricted Hamilton's principle  providing the dynamics of Lagrangian systems subject to  fractional damping \eqref{IntEqCont} (continuous setting), \eqref{IntEqDisc} (discrete setting),   as sufficient conditions for the extremals of the action. The discrete dynamics \eqref{IntEqDisc} is the result of the discretisation of the mentioned restricted Hamilton's principle (instead of the discretisation of the equations \eqref{IntEqCont} themselves), following the spirit of discrete mechanics and variational integrators \cite{MaWe}. As it is well-know, the variational principles and preservation 
properties (symplecticity, Noether's theorem \cite{AbMa}) of the generated dynamics are closely related; we find a particular example in our approach, say the preservation of the presymplectic form $\Omega$ \eqref{TwoForm}, as proven in Proposition \ref{GeomPreservation}. Nevertheless, the dynamical importance of $\Omega$ has to be clarified in future work. This two-form is defined on the fractional state space, Definition \ref{FracPhSp}, which is designed to accomodate the fractional derivatives. It is a vector bundle over the real space, particularly $\R^d\times\R^d$, which is necessary due to the unclear unique definition of fractional derivatives \eqref{FracDer} on a general smooth manifold $Q$. From the geometrical perspective, an interesting challenge for future work is to carry out this generalisation.

The discretisation of the restricted fractional principle leads to the discrete equations \eqref{IntEqDisc}, and to what we denote Fractional Variational Integrators FVIs, via the Algorithm \ref{FractionalAlgorithm}. When $\alpha=1/2$, the theoretical local truncation order of FVIs is 1 \cite{JiObNum}, which is consistent with the observed global convergence order, i.e. $O(h^{0.94})$, in Figure \ref{PositionFigure}, Lower-Right plot. However, this global convergence and the local truncation order 1 represents an improvement from what is expected from the order theorems in \cite{MaWe,PaCu}, since, as proven in \cite{FracOrder}, $\Delta^{\alpha}_{-}x_k$ \eqref{DiscFracDerDef:1}  is only a consistent approximation of $D^{\alpha}_{-}x(t_k)$ \eqref{FracDer:RL}, and thus one would expect a slower global convergence. This is an interesting phenomenom to explore in future works. Moreover, the performance of FVI is proven to be superior to other order 1 methods, such as both Euler's. This is particularly evident with respect to the energy, as shown in Figure \ref{EnergyFigure}, accounting for another example of the superior performance of integrators with variational origin in this aspect. In the context of this work, this is furthermore explained thanks to the relationship between the FVI and the Forced Variational Integrators, Algorithm \ref{ForcedAlgorithm}, obtained from the discrete Lagrange-d'Alembert principle, which is established in Corollary \ref{Corolario}. Naturally, we also test the FVI for $\alpha\neq 1/2$.

We notice that the chosen discretisations limit the local truncation order of the FVIs and their global convergence. Thus, as a natural extension of this work, we intend to carry out the application of higher-order techniques, as introduced in \cite{O14,SinaSaake}.

\section*{Appendix: Existence and uniqueness of solutions of fractional differential equations \eqref{ResFracELDamped} }

We study the existence and uniqueness of solutions of \eqref{ResFracELDamped}. It has been proven, Proposition \ref{InvTime}, that $x$ and $y$ systems are equivalent in reversed directions of time. Thus, we focus on the $x$-system and set $d=1$ for simplicity. We are considering $L:T\R\Flder\R$ a $C^2$ function; furthermore we shall consider $\lp\der^2L/\der\dot x\der\dot x\rp^{-1}$  smooth. All in all, \eqref{ResFracELDamped:a}, can be expressed as
\begin{equation}\label{System}
\begin{split}
\dot x&=v,\\
\dot v&=f(x,v)+\bar\rho(x,v)D^{\beta}_{-}x,
\end{split}
\end{equation}
with initial condition $x(t_0)=x_0,\,v(t_0)=v_0 $, $\beta=2\alpha$ and $f,\bar\rho:\R^2\Flder\R$  given by 
\[
f(x,v)=\lp\frac{\der^2L}{\der v\der v}\rp^{-1}\lp-\frac{\der^2L}{\der x\der v}v+\frac{\der L}{\der x}\rp,\quad\quad \bar\rho(x,v)=-\lp\frac{\der^2L}{\der v\der v}\rp^{-1}\rho,
\]
with $\rho\in\R_+$, adding up for the vector field $F:\R^2\Flder\R^2$
\begin{equation}\label{VectorField}
F(x,y):=\lp v\,\,,\,\, f(x,v)+\bar\rho(x,v)D^{\beta}_{-}x\rp.
\end{equation}
In proving the existence and uniqueness of solutions of \eqref{System}, we shall take a local approach; in particular we consider the set $\mathcal{B}=I_t\times I_x\times I_v\subset \R^3$, with $I_t=[t_0-\delta_t, t_0+\delta_t],\,I_x=[x_0-\delta_x, x_0+\delta_x],\, I_v=[v_0-\delta_v,v_0+\delta_v]$, $I_t\subset[a,b]$ and $\delta_t,\delta_x,\delta_v\in\R_+$. We consider $\R^d$ as a Banach space equipped with the norm 
\[
||\mathfrak{v}||=\mbox{max}\lc |\mathfrak{v}_1|,|\mathfrak{v}_2|,...,|\mathfrak{v}_d|\rc,\quad \mathfrak{v}\in\R^d.
\]
Given that, we define the cylinder $\mathcal{C}=I_t\times \mathcal{D}$, with $\mathcal{D}=\lc(x,v)\in\R^2\,|\,||(x,v)-(x_0,v_0)||\leq \mathfrak{b}\rc$, $\mathfrak{b}=||(\delta_x,\delta_v)||.$ We establish the following hypotheses:
\begin{itemize}
\item[H1.] $||f(x,v)-f(\tilde x,\tilde v)||\leq K||(x,v)-(\tilde x,\tilde v)||$ in $\mathcal{D}$ for $K\in\R_+$.
\item[H2.] $\bar\rho$ is continuous in $\mathcal{D}$.
\end{itemize}
Note that these two hypotheses follow directly from the assumptions over $L$ and $\lp\der^2L/\der\dot x\der\dot x\rp^{-1}$, say they are $C^2$ and smooth, respectively. Given this, our strategy is to prove that the vector field \eqref{VectorField} satisfy the required conditions to apply both Peano and Picard-Lindel\"of theorems  \cite{BookODEs}, ensuring the existence and uniqueness of solutions of \eqref{System} in $\tilde{\mathcal{C}}=\tilde I_t\times \mathcal{D}$, with $\tilde I_t=[t_0-\tilde\delta_t,t_0+\tilde\delta_t]$, $\tilde\delta_t=\,$min$\lc\delta_t\,,\,\mathfrak{b}/M\rc$, where $M\in\R_+$ is constructed in the proof below. Moreover, this solution $(x(t),v(t))$ can be set smooth by establishing the proper Banach space of functions when applying the aforementioned theorems.
\begin{proposition}
Given the hypotheses {\rm H1, H2}, the following is true: $F:\mathcal{D}\Flder\R^2$
\begin{enumerate}
\item   is bounded.
\item  is Lipschitz continuous.
\end{enumerate}
\end{proposition}
\begin{proof}

(1) Let $(x,v)\in \mathcal{D}$, then $||F(x,v)||=$max $\lc |v|, |f(x,v)+\bar\rho(x,v)D^{\beta}_{-}x|\rc$. On the one hand, $|v|\leq |v_0+\mathfrak{b}|<\infty$. On the other, by H1 $f$ is continuous in $\mathcal{D}$ and therefore  $|f(x,v)|\leq M_f<\infty$. By H2, $\bar\rho$ is also continuous; thus $|\bar\rho(x,v)|\leq M_{\bar\rho}<\infty$. Now, we shall use the Caputo definition \eqref{FracDer:Ca} of the fractional derivative in \eqref{VectorField}, just by using the relationship \eqref{Relation} and setting $x_0=0.$ With that, in $I_t\times\mathcal{D}$ we have
\[
\begin{split}
&|f(x,v)+\bar\rho(x,v)D^{\beta}_{-}x|\leq|f(x,v)|+|\bar\rho(x,v)||D^{\beta}_{-}x|\leq M_{f}+M_{\bar\rho}|D^{\beta}_{-}x|\\
=&M_{f}+\frac{M_{\bar\rho}}{\Gamma(1-\beta)}\Big|\int_{t_0}^t(t-\tau)^{-\beta}\dot x(\tau)d\tau\Big|=M_{f}+\frac{M_{\bar\rho}}{\Gamma(1-\beta)}\Big|\int_{t_0}^t(t-\tau)^{-\beta}v(\tau)d\tau\Big|\\
\leq & M_{f}+\frac{M_{\bar\rho}}{\Gamma(1-\beta)}\int_{t_0}^t|(t-\tau)^{-\beta}| |v(\tau)|d\tau\leq M_{f}+\frac{M_{\bar\rho}|v_0+\mathfrak{b}|}{\Gamma(1-\beta)}\int_{t_0}^t|(t-\tau)^{-\beta}|d\tau\\
\leq & M_{f}+\frac{M_{\bar\rho}|v_0+\mathfrak{b}|}{\Gamma(1-\beta)}\delta_t^{1-\beta}<\infty,
\end{split}
\]
where, according to $\beta=2\alpha$, we consider $\alpha\in[0,1/2]$, since $\beta$ must be less of equal to 1\footnote{Observe that this limits the original range of $\alpha$, i.e.~ $\alpha\in[0,1]$. Allowing that, the displayed proof of existence and uniqueness is not valid as long as we admit $\delta_t\Flder 0$, which does not mean that such existence and uniqueness of solutions may be proven in a different way.}. From this, it follows that $||F(x,v)||\leq M<\infty,$ where $M=\lc |v_0+\mathfrak{b}|\,\,\mbox{or}\,\,M_{f}+\frac{M_{\bar\rho}|v_0+\mathfrak{b}|}{\Gamma(1-\beta)}\delta_t^{1-\beta}\rc$, depending on whether the max $\lc |v|, |f(x,v)+\bar\rho(x,v)D^{\beta}_{-}x|\rc$ is achieved in the first or second entry.
\\

(2) We have that
\begin{equation}\label{Norm}
||F(x,v)-F(\tilde x,\tilde v)||=||\big( v-\tilde v,\,f(x,v)-f(\tilde x,\tilde v)+\bar\rho(x,v)D^{\beta}_{-}x-\bar\rho(\tilde x,\tilde v)D^{\beta}_{-}\tilde x\big)||,
\end{equation}
which is equal to the maximum of the absolute value of both entries. On the one hand 
\begin{equation}\label{ByConstr}
|v-\tilde v|\leq ||(x,y)-(\tilde x,\tilde v)||, 
\end{equation}
by construction. On the other
\[
\begin{split}
&|f(x,v)-f(\tilde x,\tilde v)+\bar\rho(x,v)D^{\beta}_{-}x-\bar\rho(\tilde x,\tilde v)D^{\beta}_{-}\tilde x|\leq |f(x,v)-f(\tilde x,\tilde v)|+|\bar\rho(x,v)D^{\beta}_{-}x-\bar\rho(\tilde x,\tilde v)D^{\beta}_{-}\tilde x|\\
\leq& K||(x,y)-(\tilde x,\tilde v)||+M_{\bar\rho}|D^{\beta}_{-}x-D^{\beta}_{-}\tilde x|=K||(x,y)-(\tilde x,\tilde v)||+\frac{M_{\bar\rho}}{\Gamma(1-\beta)}\big|\int_{t_0}^{t}(t-\tau)^{-\beta}(v(\tau)-\tilde v(\tau))d\tau\big|\\
\leq&K||(x,y)-(\tilde x,\tilde v)||+\frac{M_{\bar\rho}}{\Gamma(1-\beta)}||(x,y)-(\tilde x,\tilde v)||\int_{t_0}^t|(t-\tau)^{-\beta}|d\tau\\
\leq& K||(x,y)-(\tilde x,\tilde v)||+\frac{M_{\bar\rho}\delta_t^{1-\beta}}{\Gamma(1-\beta)}||(x,y)-(\tilde x,\tilde v)||=
\lp K+\frac{M_{\bar\rho}\delta_t^{1-\beta}}{\Gamma(1-\beta)}\rp ||(x,y)-(\tilde x,\tilde v)||,
\end{split}
\]
where we have used H1, H2 and \eqref{ByConstr}. Thus, it follows that
\[
||F(x,v)-F(\tilde x,\tilde v)||\leq\tilde K\,||(x,y)-(\tilde x,\tilde v)||,
\]
where $\tilde K=\lc 1\,\mbox{or}\,K+\frac{M_{\bar\rho}\delta_t^{1-\beta}}{\Gamma(1-\beta)}\rc$, depending on whether the maximum on the right hand side of \eqref{Norm} is achieved on the first or second entry. In both cases, $\tilde K\in\R_+.$
\end{proof}

\medskip\medskip

{\bf Acknowledgments}: This work has been funded by the EPSRC project: `'Fractional Variational Integration and Optimal Control''; ref: EP/P020402/1. FJ thanks Farhang Haddad Farshi for his assistance in the design of Figure \ref{Curves}.

\end{document}